\documentclass[10pt,twocolumn,twoside]{IEEEtran}

\usepackage{amsmath}
\usepackage{amssymb}
\usepackage{amsthm}
\usepackage{mathtools}
\usepackage{graphicx}
\usepackage[linesnumbered,ruled]{algorithm2e}
\usepackage[usenames,dvipsnames]{color}
\usepackage{xcolor}
\usepackage{url}
\usepackage{float}
\usepackage[english]{babel}
\usepackage{bm}
\usepackage{tikz}
\usepackage[europeanresistors]{circuitikz}
\usetikzlibrary{matrix}
\usepackage{caption}
\usepackage{subcaption}
\newcommand{\RNum}[1]{\uppercase\expandafter{\romannumeral #1\relax}}
\newtheorem{theorem}{Theorem}[section]
\newtheorem{lemma}[theorem]{Lemma}

\newtheorem{definition}[theorem]{Definition}

\newtheorem{assumption}[theorem]{Assumption}

\newtheorem{innercustomthm}{Theorem}
\newenvironment{customthm}[1]
  {\renewcommand\theinnercustomthm{#1}\innercustomthm}
  {\endinnercustomthm}

\newtheorem{innercustomlem}{Lemma}
\newenvironment{customlem}[1]
  {\renewcommand\theinnercustomlem{#1}\innercustomlem}
  {\endinnercustomlem}

\begin{document}

\title{ Control Inversion: A Clustering-Based Method for Distributed Wide-Area Control of Power Systems
}

\author{Nan~Xue,~\IEEEmembership{Student Member,~IEEE},
        Aranya~Chakrabortty,~\IEEEmembership{Senior Member,~IEEE}

\thanks{N. Xue and A. Chakrabortty are with the Department
of Electrical and Computer Engineering, North Carolina State University, Raleigh,
NC, 27695 USA, e-mail: nxue@ncsu.edu, achakra2@ncsu.edu }
\thanks{The work is supported partly by the US National Science Foundation (NSF) under grant ECCS 1054394.}
}

\maketitle
\thispagestyle{empty}
\pagestyle{empty}

\begin{abstract}
Wide-area control (WAC) has been shown to be an effective tool for damping low-frequency oscillations in power systems. In the current state of art, WAC is challenged by two main factors - namely, scalability of design and complexity of implementation. In this paper we present a control design called control inversion that bypasses both of these challenges using the idea of clustering. The basic philosophy behind this method is to project the original power system model into a lower-dimensional state-space through clustering and aggregation of generator states, and then designing an LQR controller for the lower-dimensional model. This controller is finally projected back to the original coordinates for wide-area implementation. The main problem is, therefore, posed as finding the projection which best matches the closed-loop performance of the WAC controller with that of a reference LQR controller for damping low-frequency oscillations. We verify the effectiveness of the proposed design using the NPCC $48$-machine power system model.

\end{abstract}

\section{Introduction}

Over the past few years, the occurrence of a series of blackouts in different parts of the world has led power system utility owners to look beyond the traditional approach of controlling the grid via local feedback, and instead transition to system-wide control, often referred to as {\it wide-area control} (WAC). Several papers have been reported in the literature for WAC design \cite{hinf1}-\cite{wu}, especially for damping of electro-mechanical oscillations, but its transition to practice is still challenged by two daunting factors - namely, scalability of design and complexity of implementation. For example, conventional optimal controller such as LQR and LQG require $\mathcal{O}(n^{3})$ computational complexity ($n$ for a power system can be in the order of thousands), and usually demand all-to-all communication between every generator for implementing the feedback. To address this issue of dense communication, papers such as \cite{florian,wu} have proposed sparse optimal controllers for WAC, but the problem of scalability still remains as these controllers are optimized for the original power system model. 

To bypass these challenges, in this paper we apply a design procedure called {\it control inversion} to develop a WAC controller that admits a significantly more tractable design and simpler implementation than conventional LQR. The method involves three steps. The first step is to project the full-scale power system model with $n$ generators to a lower-dimensional state space by aggregating the generators into $r$ groups ($r\leq n$). This projection is defined by a clustering set $\mathcal{I}$ that indicates the identities of generators in the $r$ groups, and a clustering weight $w$ that decides the contribution of each generator in the aggregated model. The second step is to design an LQR controller in the lower-dimensional state-space using the aggregated model. The final step is to project this controller back to the original dimension using an inverse projection. The overall complexity of this design, thus, scales only with $r$ instead of $n$. Moreover, due to the structure of the projections, the controller naturally results in a simple two-layer hierarchical implementation strategy. The main problem, therefore, is to find $\mathcal{I}$ and $w$ such that the closed-loop performance of the proposed WAC matches that of an optimal LQR controller, which for this design is considered as a reference controller. We propose two relaxations inspired by \cite{nan}, using which this model matching problem reduces to designing $(\mathcal{I},w)$ from a quadratic optimization problem that can be constructed and solved in a numerically inexpensive way.

Preliminary results on this design were presented in our recent paper \cite{nan} for a generic LTI system. The design in this paper, however, is different than that in \cite{nan} to suit the specialties of WAC. For example, unlike \cite{nan} the definition of $\mathcal{H}_{2}$ norm for solving the model matching here is only limited to a selected frequency range that targets the suppression of low-frequency oscillations (also known as {\it inter-area oscillations}) arising from the slow electro-mechanical dynamics of the synchronous generators. This distinction results in different relaxation and solution strategies than those reported in \cite{nan}. The second constraint arises from power balance between the generators as dictated by Kirchhoff's law. This reduces to an additional consensus constraint for the LQR design. Finally, the structure of the projection matrix in this paper is defined to preserve the identity of generators with multiple states, while that in \cite{nan} only preserves a scalar state.

The remainder of the paper is organized as follows. In Section \RNum{2} we present the model of a power system, and formulate the WAC design problem. The control inversion procedure is summarized in Section \RNum{3}. The design of clustering weight $w$ is presented in Section \RNum{4}, followed by the design of clustering set $\mathcal{I}$ in Section \RNum{5}. The design is verified using the NPCC $48$-machine power system model in Section \RNum{6}. Section \RNum{7} concludes the paper.

{\bf{Notation} } The following notations will be used all throughout: $\mathfrak{i}$, imaginary unit, i.e. $\mathfrak{i}^{2} = -1$; $|m|$, absolute value of $m$; $|\mathcal{S}|_{c}$, cardinality of a set $\mathcal{S}$; $\bm{1}_{n}$, column vector of size $n$ with all $1$ entries; $I_{k}$, identity matrix of size $k$, (and the subscript is omitted if the dimension is obvious from context); $M^{T}$, $M^{*}$, transpose or conjugate transpose of a matrix $M$; $M_{ij}$, the $(i,j)^{th}$ entry of a matrix $M$; $diag(m)$, diagonal matrix with vector $m$ on its principal diagonal; $M\otimes N$, Kronecker product of $M$ and $N$; $tr(M)$, trace operation on a matrix $M$; $\| M \|_{F}$, Frobenius norm of a matrix $M$, i.e. $\| M \|_{F}=\sqrt{tr(MM^{T})}$; $ker(M)$, kernel of a matrix $M$; $\bar{\sigma}(M)$, $\bar{\lambda}(M)$, largest singular value, or eigenvalue with largest real part of a matrix $M$; $\underline{\sigma}(M)$, $\underline{\lambda}(M)$, smallest singular value, or eigenvalue with smallest real part of a matrix $M$. A transfer function matrix is defined as $g(s)=C(sI-A)^{-1}B+D$, with a realization form of $g(s)=\left[
\begin{array}{c|c}
A & B \\ \hline
C & D
\end{array}
\right]
$. 

{\bf Proofs:} We provide the proofs of all theorems stated in this paper in the Appendix unless noted otherwise.

\section{Problem Formulation}

\subsection{Power System Model}

Consider a power system network with $n+n_{l}$ number of buses. Without loss of generality, we classify the first $n$ buses to be generator buses, and the remaining $n_{l}$ buses as load buses. While several detailed dynamic models for synchronous generators exist in the literature, a convenient model that is often used for small-signal oscillation analysis and damping control is the so-called flux-decay model with a fast exciter \cite{Pai}. The model is described by the following set of differential-algebraic equations,
\begin{align}
\dot{\delta}_{i}(t) & = \Omega_{i}(t), \label{fourth1} \\
M_{i}\dot{\Omega}_{i}(t) & = P_{mi}(t) - E_{qi}^{\prime}(t)I_{qi}(t) - D_{i}\Omega_{i}(t), \label{fourth2} \\
T_{doi}^{\prime}\dot{E}_{qi}^{\prime}(t) & = - E_{qi}^{\prime}(t) - (x_{di} - x_{di}^{\prime})I_{di}(t) + E_{fdi}(t), \label{fourth3} \\
T_{Ai}\dot{E}_{fdi}(t) & = -E_{fdi}(t) + K_{Ai}(V_{\mathrm{ref,i}} - V_{i}(t)) + u_{i}(t), \label{fourth4} 
\end{align}
for $i=1,...,n$.\footnote{For ease of notation, we will omit the augment $t$ from all variables.} The state variables $(\delta_{i}, \Omega_{i}, E_{qi}^{\prime}, E_{fdi})$ are respectively the phase angle, frequency deviation from the steady-state synchronous frequency ($120\pi$ radian/second), the quadrature-axis internal voltage, and the field excitation voltage; $u_{i}$ is an excitation voltage signal which can be used as a control input. Equations (\ref{fourth1}-\ref{fourth2}) are referred to as the swing equations, and (\ref{fourth3}-\ref{fourth4}) are as the excitation equations. $M_{i}$ is the generator inertia, $P_{mi}$ is the mechanical power input from the $i^{th}$ turbine, $D_{i}$ is the generator damping factor, $T_{doi}^{\prime}$ is the direct-axis excitation time constant, $x_{di}$ is the direct-axis synchronous reactance, $x_{di}^{\prime}$ is the direct-axis transient reactance, $I_{qi}$ and $I_{di}$ together denote the current flow $(I_{qi} - \mathfrak{i}I_{di} )e^{\mathfrak{i}\delta_{i}}$ from the generator to the terminal bus, $V_{i}e^{\mathfrak{i}\theta_{i}}$ is the voltage phasor at the $i^{th}$ bus, $V_{\mathrm{ref},i}$ is the set point value of the generator bus voltage, $T_{Ai}$ is the regulator time constant, and $K_{Ai}$ is the regulator gain. For the purpose of WAC, we consider $P_{mi}$ to be constant, and design controller using only the excitation voltage $u_{i}$. $I_{qi}$, $I_{di}$, $V_{i}$, and $\theta_{i}$ are algebraic variables that can be eliminated from (\ref{fourth1}-\ref{fourth4}) by expressing them in terms of $(E_{qi},\delta_{i})$, $i=1,..,n$, using power balance equations through a process called Kron-reduction \cite{Pai}. The resulting $4n$ nonlinear equations can, thereafter, be used to determine the steady-state equilibrium $(\delta_{i0},\Omega_{i0},E_{qi0}^{\prime},E_{fdi0})$, $i=1,...,n$. Considering a small-signal perturbation around this equilibrium point, the small-signal model for the power system network can finally be derived as
\begin{align}
\begin{bmatrix}
\Delta \dot{\delta} \\
M \Delta \dot{\Omega} \\
T_{do}^{\prime} \Delta \dot{E}_{q}^{\prime} \\
T_{A} \Delta \dot{E}_{fd}
\end{bmatrix} {=} \begin{bmatrix}
0 & I & 0 & 0 \\
L_{1} & -D & F_{1} & 0 \\
L_{2} & 0 & F_{2} & I \\
L_{3} & 0 & F_{3} & -I
\end{bmatrix} \begin{bmatrix}
\Delta \delta \\
\Delta \Omega \\
\Delta E_{q}^{\prime} \\
\Delta E_{fd}
\end{bmatrix} {+} \begin{bmatrix}
0 \\ 0 \\ 0 \\ I
\end{bmatrix}\Delta u, \label{psf}
\end{align}
where $\Delta \delta = [ \Delta \delta_{1} \cdots  \Delta \delta_{n}]^{T}$, $\Delta \Omega = [\Delta \Omega_{1} \cdots \Delta \Omega_{n}]^{T}$, $\Delta E^{\prime}_{q} = [\Delta E^{\prime}_{q1}  \cdots \Delta E^{\prime}_{qn}]^{T}$, $\Delta E_{fd} = [\Delta E_{fd1} \cdots \Delta E_{fdn}]^{T}$, and $\Delta u = [\Delta u_{1} \cdots \Delta u_{n}]^{T}$ are the vectors of states and input, and diagonal matrices $M=diag(M_{1},...,M_{n})$, $T_{do}^{\prime}=diag(T_{do1}^{\prime},...,T_{don}^{\prime})$, $T_{A}=diag(T_{A1},...,T_{An})$, and $D=diag(D_{1},...,D_{n})$. The expressions for the submatrices inside the state matrix follow from linearization, and are provided in Appendix A. It can be easily shown that matrices $L_{1}$, $L_{2}$ and $L_{3}$ are asymmetric Laplacian matrices with zero row sums, and matrices $F_{1}$, $F_{2}$ and $F_{3}$ are diagonally dominant. For ease of analysis, we further apply a coordinate transformation on (\ref{psf}) using $x=(I_{4}\otimes M^{\frac{1}{2}})[\Delta \delta\ \Delta \Omega\ \Delta E_{q}^{\prime}\ \Delta E_{fd}]^{T}$. The small-signal model (\ref{psf}) can then be transformed into
\begin{align}
\begin{bmatrix}
\dot{x}_{1} \\
\dot{x}_{2} \\
\dot{x}_{3} \\
\dot{x}_{4}
\end{bmatrix} {=} \underbrace{\begin{bmatrix}
0 & I & 0 & 0 \\
L_{1m} & -D_{m} & F_{1m} & 0 \\
L_{2m} & 0 & F_{2m} & I \\
L_{3m} & 0 & F_{3m} & -I
\end{bmatrix} }_{A} \underbrace{\begin{bmatrix}
x_{1} \\
x_{2} \\
x_{3} \\
x_{4}
\end{bmatrix} }_{x} {+} \underbrace{\begin{bmatrix}
0 \\ 
0 \\
0 \\
B_{1} 
\end{bmatrix} }_{B} u {+} B_{d}d.
\label{psff}
\end{align}
Note that in (\ref{psff}) we have also added an extra term $B_{d}d$, where $d\in \mathbb{R}^{n_{d}}$ represents a disturbance entering into the system through matrix $B_{d}\in \mathbb{R}^{n\times n_{d}}$. The remaining matrices are defined by %
{\small
\begin{align*}
& B_{1} = M^{\frac{1}{2}}T_{A}^{-1},\ L_{1m} = M^{-\frac{1}{2}}L_{1}M^{-\frac{1}{2}},\ L_{2m} = M^{\frac{1}{2}}T_{do}^{\prime-1}L_{2}M^{-\frac{1}{2}},\\
& L_{3m} {=} M^{\frac{1}{2}}T_{A}^{{-}1}L_{3}M^{{-}\frac{1}{2}},\ D_{m} {=} M^{{-}\frac{1}{2}}DM^{{-}\frac{1}{2}}, F_{1m} {=} M^{{-}\frac{1}{2}}F_{1}M^{{-}\frac{1}{2}}, \\
& F_{2m} = M^{\frac{1}{2}}T_{do}^{\prime-1}F_{2}M^{-\frac{1}{2}},\quad F_{3m} = M^{\frac{1}{2}}T_{A}^{-1}F_{3}M^{-\frac{1}{2}}.
\end{align*}}%
Equation (\ref{psff}) will be used as the power system model for our proposed WAC design.

\subsection{Wide-Area Control}

The objective of wide-area control is to improve the transient performance of the power system model (\ref{psff}), especially in enhancing the damping of the complex eigenvalues of $A$ whose frequencies lie in the inter-area frequency range (typically from $0.1$ Hz to $2$ Hz). This problem is posed as a standard LQR optimal control problem. Given two real-valued design matrices $Q \succeq 0$ and $R \succ 0$, the goal is to design $u(t)=-Kx(t)$ that minimizes
the cost function
\begin{align}
J := \int_{0}^{\infty} [ x^{T}(t)Qx(t) + u^{T}(t)Ru(t) ] \mathrm{d}t.
\label{lqr}
\end{align}
We assume Phasor Measurement Units (PMUs) to be installed at a geometrically observable set of buses in the network so that all the generator voltage phasors $(V_{i}, \theta_{i})$ and currents $(I_{qi}, I_{di})$, $i=1,...,n$ can be computed from these measurements, followed by decentralized estimation of the generator states using unscented Kalman filters (for details of this state estimation, please see \cite{Pal}). The state $x$ is, therefore, assumed to be known for implementing the controller. The details of this implementation will be amplified more in the next section. 

Solving (\ref{lqr}), however, is subject to $\mathcal{O}(n^{3})$ computational complexity, and the resulting feedback matrix $K$ is usually a dense matrix, which necessitates an all-to-all communication between all generators for implementing the feedback. Since in any real power system $n$ can be easily in the order of hundreds to thousands, the design soon becomes unscalable. Therefore, instead of applying an optimal LQR controller for WAC, in this paper we resort to a sub-optimal controller $u=-\hat{K}x$ that can potentially bypass these challenges. The controller $\hat{K}$ is supposed to emulate the optimal controller $K$ in terms of their closed-loop responses defined as follows.

{\bf Performance metric:} The performance metric for evaluating WAC is defined as the small-signal power flow between any pair of generators, or equivalently the difference of their phase angles, and the small-signal generator frequencies. We write this as $y=Cx$ where
\begin{align}
C = \begin{bmatrix}
\bar{C} & 0 & 0 & 0 \\
0 & I_{n} & 0 & 0
\end{bmatrix}(I_{4}\otimes M^{-\frac{1}{2}}).
\label{perfc}
\end{align} 
In (\ref{perfc}) $\bar{C} \in \mathbb{R}^{n_{\delta}\times n}$ is an indicator matrix with all zeroes except $\bar{C}_{ki} = 1$, $\bar{C}_{kj} = -1$, $i,j\in \{ 1,...,n\}$, $k=1,...,n_{\delta}$. The output $y$ so defined measures $n_{\delta}$ pairs of angle differences between any chosen pair of generators, and the frequency deviations of all generators. Using this definition, we consider two transfer function matrices (TFM) from the disturbance input $d$ to the performance output $y$, which are written respectively as - (1) closed-loop system with optimal controller:
\begin{align}
g(s) = C(sI - A + BK)^{-1}B_{d}, \label{bmtf}
\end{align} 
(2) closed-loop system with proposed sub-optimal controller: 
\begin{align}
\hat{g}(s) = C(sI - A + B\hat{K})^{-1}B_{d}. \label{sotf}
\end{align} 
Note that the disturbance $d$ in this case is a design metric for evaluating the dynamic response of the swing states. Without loss of generality, we assume $(A,B_{d})$ to be controllable. In these notations, the main problem of interest is stated next.

{\bf Problem statement:} Given TFMs $g(s)$ and $\hat{g}(s)$, find a sub-optimal LQR controller $u = - \hat{K}x$ that solves the WAC model matching problem:
\begin{equation}
\begin{aligned}
& \underset{\hat{K}}{\mathrm{minimize}}
& &  \| g(s) - \hat{g}(s) \|_{\mathcal{H}_{2},\bar{\omega}}
\end{aligned}, \label{main} \tag{WM}
\end{equation}
where the norm $\| \cdot \|_{\mathcal{H}_{2},\bar{\omega}}$ is defined by
\begin{align}
\|h(s)\|_{\mathcal{H}_{2},\bar{\omega}} = \sqrt{\frac{1}{2\pi}\int_{-\bar{\omega}}^{\bar{\omega}}tr[h^{*}(j\omega)h(j\omega)]\mathrm{d}\omega}, \label{h2w}
\end{align}
for any stable transfer function matrix $h(s)$, and $[0,\bar{\omega}]$, $\bar{\omega}\in \mathbb{R}$ indicating the frequency range of inter-area oscillations in the power system model (\ref{psff}). The controller $\hat{K}$ should satisfy the following three constraints.
\begin{enumerate}
\item {\it Consensus} - Since the total amount of power in the network remains conserved, the model (\ref{psff}) exhibits a consensus property which manifests as a zero eigenvalue in the state matrix $A$. The same property must also be true in closed-loop, i.e. $A -B\hat{K}$ must have a zero eigenvalue (often referred to as the DC mode \cite{Pai}).
 
\item {\it Computation} - The design complexity of $\hat{K}$ should be less than $\mathcal{O}(n^{3})$.
\item {\it Implementation} - The structure of $\hat{K}$ is desired to produce a much simpler communication topology between the generators.
\end{enumerate}

In order to solve (\ref{main}) under these three constraints, we employ a design procedure called {\it control inversion}. The control inversion strategy was introduced in our recent work \cite{nan} for a generic LTI system. To cope with the specific properties and constraints that arise from the power system model (\ref{psff}), this paper develops three major extensions over \cite{nan} - namely the consensus constraint listed above, the structural constraint on $\hat{K}$ which now preserves the identity of generators with all four states instead of the scalar state assumption in \cite{nan}, and finally defining the $\mathcal{H}_{2}$ norm in (\ref{main}) over the inter-area frequency range using (\ref{h2w}) instead of the standard $\mathcal{H}_{2}$ norm definition in \cite{nan}. We next provide an overview of this control inversion strategy.

\section{Overview of Control Inversion}

\begin{figure*}
\centering
\includegraphics[width=1\textwidth]{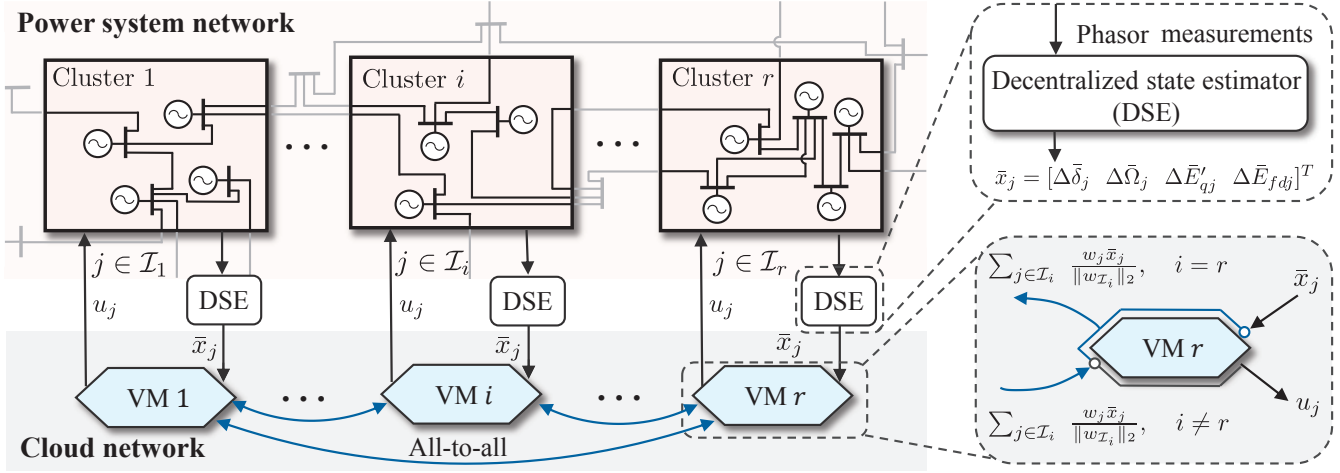} 
\caption{Cyber-physical implementation of proposed wide-area controller $\hat{K}$}
\label{arch}
\vspace{-1em}
\end{figure*}

Control inversion starts from defining a structured projection matrix based on clustering of the $n$ generators, as follows. 
\begin{definition}
Given a vector $w\in \mathbb{R}^{n}$ and an integer $r$, where $0< r \leq n$, define $r$ non-empty, distinct, and non-overlapping subsets of the generator index set $\mathcal{V} = \{ 1,...,n \}$, respectively denoted as $\mathcal{I} = \{ \mathcal{I}_{1},...,\mathcal{I}_{r} \}$, such that $\mathcal{I}_{1} \cup ... \cup \mathcal{I}_{r} = \mathcal{V}$. We refer to such a set $\mathcal{I}$ as the clustering set. A clustering-based projection matrix $P \in \mathbb{R}^{r\times n}$ is defined as 
\begin{align}
P_{ij} := \begin{cases}
\frac{w_{j}}{| \| w_{\mathcal{I}_{i}} \|_{2} }, \quad j\in \mathcal{I}_{i} \\
0,\quad \text{otherwise} 
\end{cases}, 
\label{Pdefe}
\end{align}
where $w_{\mathcal{I}_{i}} = [w_{\mathcal{I}_{i}\{ 1 \}},\cdots,w_{\mathcal{I}_{i}\{ | \mathcal{I}_{i} |_{c} \}}]^{T} $ is a non-zero vector, and $\mathcal{I}_{i}\{ j \}$ denotes the $j^{th}$ element in the set $\mathcal{I}_{i}$. 
\label{Pdef}
\end{definition}

With the projection $P$ defined above, the design for $\hat{K}$ can be summarized through the following three steps.

\subsection{Design Strategy for $\hat{K}$ }

\subsubsection{Projection to lower-dimensional space}
Define a stacked projection matrix for the power system model (\ref{psff}) as
\begin{align}
\Pi = I_{4}\otimes P.
\label{Pstack}
\end{align}
With $\Pi$, a lower-dimensional system can be defined as 
\begin{align}
\begin{bmatrix}
\dot{\tilde{x}}_{1} \\
\dot{\tilde{x}}_{2} \\
\dot{\tilde{x}}_{3} \\
\dot{\tilde{x}}_{4}
\end{bmatrix} {=} \underbrace{\begin{bmatrix}
0 & I & 0 & 0 \\
\tilde{L}_{1m} & -\tilde{D}_{m} & \tilde{F}_{1m} & 0 \\
\tilde{L}_{2m} & 0 & \tilde{F}_{2m} & I \\
\tilde{L}_{3m} & 0 & \tilde{F}_{3m} & -I
\end{bmatrix} }_{\tilde{A} = \Pi A \Pi^{T}} \underbrace{\begin{bmatrix}
\tilde{x}_{1} \\
\tilde{x}_{2} \\
\tilde{x}_{3} \\
\tilde{x}_{4}
\end{bmatrix} }_{\tilde{x}} {+} \underbrace{\begin{bmatrix}
0 \\ 
0 \\
0 \\
\tilde{B}_{1} 
\end{bmatrix} }_{\tilde{B}=\Pi B } \tilde{u} {+} \tilde{B}_{d}d
\label{psr}
\end{align}
with $\tilde{B}_{1} = P B_{1}$, $\tilde{B}_{d} = \Pi B_{d}$, $\tilde{D}_{m} = P D_{m} P^{T}$, $\tilde{L}_{im} = P L_{im} P^{T}$, and $\tilde{F}_{im} = P F_{im} P^{T}$, $i=1,2,3$. An important point to note is that unlike $x$ in (\ref{psff}) the state vector $\tilde{x}$ in (\ref{psr}) does not have any physical meaning. The model in (\ref{psr}) is a hypothetically defined model that is only meant to facilitate the design of $\hat{K}$.

\subsubsection{Lower-dimensional design}
Based on (\ref{psr}), we next pose a lower-dimensional LQR problem with respect to the two projected design parameters $\tilde{Q} = \Pi Q\Pi^{T} \in \mathbb{R}^{4r\times 4r}$ and $\tilde{R}=R$. This LQR problem is approached by the consensus-preserving reformulation, and yields a lower-dimensional matrix $\tilde{X}$ from 
\begin{align}
\tilde{X} = \mathrm{CPLQR}(\tilde{A},\tilde{B},\tilde{Q},\tilde{R}). 
\label{cplqro}
\end{align}
where the definition of the function $\mathrm{CPLQR}(\cdot)$ will be explained in Section \RNum{4}.

\subsubsection{Inverse projection to original coordinates}
Once $\tilde{X}$ is solved from (\ref{cplqro}), one can project it back to the original coordinates via inverse projection
\begin{align}
\hat{X} = \Pi^{T} \tilde{X} \Pi.
\end{align}
This projected matrix $\hat{X}$ can be implemented in (\ref{psff}) using $u=-R^{-1}B^{T}\hat{X}x$. The sub-optimal controller for (\ref{main}), therefore, follows as 
\begin{align}
\hat{K} = R^{-1}B^{T}\hat{X}. \label{Khat}
\end{align}

Equations (\ref{Pdefe}) to (\ref{Khat}) define the control inversion method, which reduces to finding the clustering set $\mathcal{I}$ and clustering weight $w$ to solve (\ref{main}). The benefit here is that $\hat{K}$ is designed using (\ref{cplqro}), which involves matrices of dimension $r\leq n$. If the system operator chooses $r$ to be sufficiently small (e.g. there can be close to $n=500$ generators but only $r=5$ clusters), the design of $\hat{K}$ becomes numerically more tractable than an LQR controller $K$ in $n$-dimension. Moreover, the controller $\hat{K}$ is naturally imposed with the structure of $\Pi$, which results in a simple two-layer hierarchical implementation scheme, as described next.

\subsection{Implementation Strategy for $\hat{K}$}
In a practical power system all four states of a generator may not be directly measurable. One plausible way of estimating the states can be through the decentralized state estimator (DSE) that has been recently reported in \cite{Pal}. Here we denote by $\bar{x}$ the estimated state vector for $x$. The corresponding implementation architecture of this scheme is shown in Fig. \ref{arch}. The generators are divided into $r$ clusters, each equipped with its own DSE. Each cluster is assumed to have PMUs placed such that they make the generator buses geometrically observable. The voltage and current phasors of every generator buses are computed from these PMU measurements, and sent to the DSE of that cluster. The $i^{th}$ DSE generates the state estimates $\bar{x}_{j} = [\bar{\delta}_{j}\ \bar{\Omega}_{j}\ \bar{E}_{qj}^{\prime}\ \bar{E}_{fdj}]^{T}$, $j \in \mathcal{I}_{i}$, $i=1,...,r$. The state estimates are transmitted to $r$ distributed computers (referred to as virtual machine or VMs in Fig. \ref{arch}) that can be created in a cloud network \cite{aranyab}. The implementation of the feedback $u = -\hat{K}\bar{x}$ follows three steps:

{\it Step 1 - state averaging $\Pi \bar{x}$:}  First, the $i^{th}$ VM receives all the $\bar{x}_{j}$ from its designated DSE, i.e. $j\in \mathcal{I}_{i}$, $i=1,...,r$. Each VM then computes the weighted averaged state vector $\sum_{j\in \mathcal{I}_{i}} \frac{w_{j}\bar{x}_{j}}{ \| w_{\mathcal{I}_{i}} \|_{2} }$ for its cluster, $i=1,...,r$. This averaged vector corresponds to the $(i,i+r,i+2r,i+3r)^{th}$ entries of the vector $\Pi \bar{x}$. 

{\it Step 2 - lower-dimensional feedback $\tilde{X}\Pi \bar{x}:$} Next, the VMs exchange their weighted averages, and each computes the $4r$-dimensional vector $\tilde{X}\Pi \bar{x}$. Note that no VM will be able to infer individual state measurements from other clusters, and hence data privacy between the VMs is maintained.

{\it Step 3 - broadcast of control $u = -R^{-1}B^{T}\hat{X}x$:} Finally, the $i^{th}$ VM computes the control signal $u_{j}$, $j\in \mathcal{I}_{i}$ by taking linear combinations of the elements in $\tilde{X}\Pi \bar{x}$. The linear combination follows directly from $u = -(R^{-1}B^{T}\Pi^{T})\tilde{X}\Pi \bar{x}$. The control signal $u_{j}$, $j=1,...,n$ is then transmitted to its respective generator.

In the worst case when every generator is equipped with a PMU, the hierarchical implementation results in at most $n+{{r}\choose{2} }$ bidirectional communication links, including $n$ links between PMUs and DSEs, and ${{r}\choose{2} }$ links between VMs assuming the DSEs to be located directly inside the VMs. If $r\ll n$, this communication topology can be significantly sparser than that of an optimal LQR controller which requires ${{n}\choose{2}}$ number of links. Moreover, compared to the sparsity-based designs in \cite{florian} that exploit the controller structure by imposing $l_{1}$ penalties in the objective function, the structure in $\hat{K}$ instead is parameterized by the clustering set $\mathcal{I}$ and weight $w$, and therefore, can be flexibly tuned and designed. In the next two sections we present the design of these two parameters under the consensus constraint and the computational preference listed under (\ref{main}).

\section{Consensus-Preserving Reformulation}

The standard LQR formulation for WAC becomes infeasible when one imposes the consensus constraint. To resolve this problem, in this section we propose a reformulation of the standard LQR, referring it as consensus-preserving LQR (CPLQR). We start by explaining the consensus property of the power system model (\ref{psff}).

\subsection{Consensus Property of Power System}

The consensus behavior of the model (\ref{psff}) is decided by the three asymmetric Laplacian matrices $L_{1}$, $L_{2}$ and $L_{3}$ defined in (\ref{psf}). Since the rows of each of these matrices sum to zero, they have at least one zero eigenvalue which forces the states to reach a consensus value. We characterize this consensus property of (\ref{psff}) as follows.

\begin{definition}
The power system model (\ref{psff}) admits an angular consensus point, which is defined by a zero eigenvalue of $A$ and its right eigenvector $v_{0}$ as
\begin{align}
v_{0} = \begin{bmatrix}
\bar{v}^{T} & 0 & 0 & 0
\end{bmatrix}^{T},\quad \bar{v} = \frac{M^{\frac{1}{2}}}{\sqrt{tr(M)}} \mathbf{1}_{n}.
\label{vbar}
\end{align}
Here $\bar{v}$ is the right null space of $L_{1m}$, $L_{2m}$ and $L_{3m}$, i.e. $L_{1m}\bar{v} = L_{2m}\bar{v} = L_{3m}\bar{v} = 0$. \label{consensus}
\end{definition} 

The physical interpretation of consensus lies in the phase angle $\Delta \delta$, which corresponds to the non-zero entries in $v_{0}$. From \cite{Pai}, the small-signal power flow between generators $i$ and $j$ is directly proportional to the angle difference $\Delta \delta_{i} - \Delta \delta_{j}$. One immediate consequence of this property is the non-uniqueness of equilibrium value of the power flow. That is, both $(\delta_{i0},\delta_{j0})$ and $(\delta_{i0} + \Delta \delta_{i},\delta_{j0} + \Delta \delta_{j})$ will result in the same equilibrium for any angle deviations as long as $\Delta \delta_{i} - \Delta \delta_{j} = 0$. Due to this consensus behavior, we define the following stability criterion.
\begin{definition}
{\bf (Consensus stability):} The power system model (\ref{psff}) is called consensus stable if all eigenvalues except for one zero eigenvalue of $A$ lie in the left half plane. 
\label{constab}
\end{definition}

Consensus stability is basically a relaxation of asymptotic stability of (\ref{psff}) with the consensus point excluded. In practice, the power flows in a power system will always remain balanced, and thereby preserve angular consensus. Hence, we conform to the following assumption throughout the paper.
\begin{assumption}
The power system network model (\ref{psff}) is consensus stable.
\end{assumption}

In the existing literature several papers such as \cite{florian} have proposed control designs that neglect the consensus property of power system models. The flip side of these designs is that the control will force all angle deviations $\Delta \delta_{i}$ to converge to zero. In reality, however, it may be preferable to drive this angle deviation to a nearby consensus value, e.g. $\Delta \delta_{i} = \frac{1}{n} \sum_{j=1}^{n}\Delta \delta_{j}$, $i=1,...,n$, especially if $\Delta \delta_{i}$ has large absolute magnitude. Note that it is also possible to get rid of the consensus point in (\ref{psff}) by modeling the states $\Delta \delta$ directly as angular differences with respect to a reference generator \cite{Pai}, or similarly by applying an orthonormal projection on (\ref{psff}) as shown in \cite{wu}. The drawback, however, is that the states in these models no longer retain their individual identities, as a result of which the network structure of $A$ is destroyed. For our design of $\hat{K}$, we, therefore, stick to the notion of consensus stability, and construct $\hat{K}$ such that angular consensus is preserved in the closed-loop state matrix $A+B\hat{K}$. Before proceeding to the reformulation of LQR, we make an additional assumption to ensure that the model (\ref{psff}) is feasible for control. 
\begin{assumption}
Matrix $F_{1}$ (or $F_{1m}$) is nonsingular at the equilibrium $(\delta_{i0},\Omega_{i0},E_{qi0}^{\prime},E_{fdi0})$, $i=1,\cdots,n$.
\label{invt}
\end{assumption}
This assumption holds in practice because $F_{1}$ is a diagonal-dominant matrix. This results in structural controllability for both the original system (\ref{psff}) and lower-dimensional system (\ref{psr}) as follows. 
\begin{lemma}
The pairs $(A,B)$ and $(\tilde{A},\tilde{B})$ are controllable.
\label{cont}
\end{lemma}

\subsection{Consensus-Preserving LQR (CPLQR)}

Recall the standard LQR problem (\ref{lqr}). The optimal solution for (\ref{lqr}) is associated with the algebraic Riccati equation (ARE)
\begin{align}
A^{T}X + XA + Q - XGX = 0,
\label{are}
\end{align}
where $G=BR^{-1}B^{T}$. According to \cite{rc}, the ARE (\ref{are}) admits a unique stabilizing solution $X=X^{T} \succeq 0$ if $(A,BR^{-\frac{1}{2}})$ is stabilizable, and $(Q^{\frac{T}{2}},A)$ is detectable. Given such a solution $X$, the optimal feedback matrix can then be found by $K = R^{-1}B^{T}X$. Here, $X$ guarantees asymptotic stability of the closed-loop system, or equivalently $A-BK$ to be Hurwitz. Incorporating a consensus constraint in this formulation, which means $A-BK$ now must have a zero eigenvalue, makes the LQR problem fundamentally ill-posed. To preserve the well-posedness of LQR, we propose its consensus-preserving reformulation as follows.
\begin{lemma}
{\bf (CPLQR)} Denote the eigenvalue decomposition of $A$ by
\begin{align}
A = V\Lambda V^{-1} = \begin{bmatrix}
v_{0} & v_{1}
\end{bmatrix} \begin{bmatrix}
0 & \\ & \Lambda_{1}
\end{bmatrix} \begin{bmatrix}
w^{T}_{0} \\ w^{T}_{1}
\end{bmatrix},
\label{evd}
\end{align}
where $v_{0}$ is as defined in Proposition \ref{consensus}. Consider an arbitrary real-valued scalar $\epsilon >0$ and define
\begin{align}
A_{\epsilon} : = A - \epsilon v_{0}w_{0}^{T}. \label{Aepsilon}
\end{align}
Suppose the only null space of $Q\succeq 0$ is at $Qv_{0}=0$. The LQR problem $(A_{\epsilon},B,Q,R)$ is guaranteed with a unique stabilizing solution $K = R^{-1}B^{T}X_{\epsilon}$ from
\begin{align}
A_{\epsilon}^{T}X_{\epsilon} + X_{\epsilon}A_{\epsilon} + Q - X_{\epsilon}GX_{\epsilon} = 0,
\label{areep}
\end{align}
and it satisfies that $X_{\epsilon}v_{0}=0$. Irrespective of $\epsilon$, closed-loop state matrix $A-BK$ preserves the angular consensus, and has all of its eigenvalues on the left half plane except for one zero eigenvalue. 
\label{dclqr}
\end{lemma}

For fair comparisons between the controller $K$ in (\ref{bmtf}) and $\hat{K}$ in (\ref{sotf}), from this point onwards we will consider the benchmark LQR design (\ref{lqr}) in terms of its CPLQR reformulation $(A_{\epsilon},B,Q,R)$. We will stick to the same choice of $Q$ as in Lemma \ref{dclqr}, and consider the optimal controller as $K=R^{-1}B^{T}X_{\epsilon}$.

\subsection{Choice of $w$}

The CPLQR reformulation enables the choice of the clustering weight $w$, and the definition of the operator $\mathrm{CPLQR}(\cdot)$ in (\ref{cplqro}) so that $\hat{K}$ bypasses the consensus constraint. The selection of $w$ is guided by the following property of matrices $\tilde{A}$ and $\tilde{Q}$ in lower dimension.
\begin{lemma}
Let $w = \bar{v}$. State matrix $\tilde{A}$ from the lower-dimensional model (\ref{psr}) preserves the zero eigenvalue of angular consensus at its right eigenvector $\tilde{v}_{0} = \Pi v_{0}$, i.e. $\tilde{A}\tilde{v}_{0}=0$. Matrix $\tilde{Q}$ is positive semi-definite, and possesses its only null space at $\tilde{Q}\tilde{v}_{0} = 0$.
\label{psrp}
\end{lemma}

From Lemma \ref{psrp}, by choosing $w=\bar{v}$ both $\tilde{A}$ and $\tilde{Q}$ inherit the null space $\tilde{v}_{0}=\Pi v_{0}$ projected from the consensus point. This satisfies the same condition required by Lemma \ref{dclqr}, and thus, allows a CPLQR reformulation for the lower-dimensional LQR problem $(\tilde{A},\tilde{B},\tilde{Q},\tilde{R})$. Denote the eigenvalue decomposition of $\tilde{A}$ by 
\begin{align}
\tilde{A} = \tilde{V}\Lambda \tilde{V}^{-1} = \begin{bmatrix}
\tilde{v}_{0} & \tilde{v}_{1}
\end{bmatrix} \begin{bmatrix}
0 & \\ & \tilde{\Lambda}_{1}
\end{bmatrix} \begin{bmatrix}
\tilde{w}_{0}^{T} \\ \tilde{w}_{1}^{T} 
\end{bmatrix}, \label{evdr}
\end{align}
and define a matrix $\tilde{A}_{\epsilon} = \tilde{A} - \epsilon \tilde{v}_{0}\tilde{w}_{0}^{T}$ for any $\epsilon >0$. The lower-dimensional matrix $\tilde{X}$ from (\ref{cplqro}), therefore, is the solution of the lower-dimensional ARE
\begin{align}
\tilde{A}_{\epsilon}^{T}\tilde{X} + \tilde{X}\tilde{A}_{\epsilon} + \tilde{Q} - \tilde{X}\tilde{G}\tilde{X} = 0. \label{arerr}
\end{align}
We denote this operation by $\tilde{X}=\mathrm{CPLQR}(\tilde{A},\tilde{B},\tilde{Q},\tilde{R})$ as in (\ref{cplqro}). Note that by definition matrix $\tilde{A}_{\epsilon}$ has the same basis of $\tilde{A}$. Given that $(\tilde{A},\tilde{B})$ is controllable from Lemma \ref{cont}, the pair $(\tilde{A}_{\epsilon},\tilde{B})$ would remain controllable. In addition, the only null space of $\tilde{Q}$ is at $\tilde{v}_{0}$, which corresponds to a stable eigenvalue $-\epsilon$ of $\tilde{A}_{\epsilon}$, and thus makes $(\tilde{Q},\tilde{A}_{\epsilon})$ detectable. These two conditions together guarantee a unique solution $\tilde{X}\succeq 0$ for (\ref{arerr}), which also satisfies $\tilde{X}\tilde{v}_{0} = 0$ according to Lemma \ref{dclqr}. Thereby, the closed-loop state matrix $A-B\hat{K} = A - G\hat{X}$ yields 
\begin{align}
(A-G\hat{X})v_{0} = G \Pi^{T}\tilde{X} \Pi v_{0} = G\Pi^{T}\tilde{X}\tilde{v}_{0} = 0. \label{e25}
\end{align}
Equation (\ref{e25}) suggests that $A - B\hat{K}$ will have a zero eigenvalue, that is $\hat{K}$ will preserve closed-loop consensus. We conclude this result with the following theorem.

\begin{theorem}
Suppose the only null space of $Q\succeq 0$ is at $Qv_{0}=0$. By choosing $w=\bar{v}$, the control inversion steps (\ref{Pdefe}-\ref{Khat}) admit a controller $\hat{K}$. Furthermore, $\hat{K}$ preserves the angular consensus in closed-loop.
\label{wsolution}
\end{theorem}
\begin{proof}
The proof follows directly from the discussions in this subsection.
\end{proof}

\section{Design of Generator Clustering Sets}

With the analytical solution of $w$ provided in Section \RNum{4}, the only unknown left for designing $\Pi$ is the clustering set $\mathcal{I}$, which dictates the implementation structure of $\hat{K}$. In this section, we present a design for $\mathcal{I}$ to solve the minimization problem (\ref{main}) under its computational constraint. We start by stating two equivalent realizations for $g(s)$ and $\hat{g}(s)$ as follows.
\begin{lemma}
The two TFMs $g(s)$ and $\hat{g}(s)$ admit the realizations $g(s) = g_{\epsilon}(s)$ and $\hat{g}(s) = \hat{g}_{\epsilon}(s)$ respectively, with
\begin{align}
g_{\epsilon}(s) = \left[
\begin{array}{c|c}
A_{\epsilon} {-} GX_{\epsilon} & B_{d} \\ \hline
C & 0
\end{array}
\right],\ \hat{g}_{\epsilon}(s) = \left[
\begin{array}{c|c}
A_{\epsilon} {-} G\hat{X} & B_{d} \\ \hline
C & 0
\end{array}
\right],
\label{tfshift}
\end{align}
and $A_{\epsilon}$ as defined in (\ref{Aepsilon}).
\label{tfequ}
\end{lemma}

The equivalencies between these TFMs can be verified using a coordinate transformation $V^{-1}$ and $V$ from (\ref{evd}). Facilitated by Lemma \ref{tfequ}, the consensus stability of $g(s)$ and $\hat{g}(s)$ as in (\ref{bmtf}) and (\ref{sotf}) simply becomes asymptotic stability of $g_{\epsilon}(s)$ and $\hat{g}_{\epsilon}(s)$ in (\ref{tfshift}). The main problem (\ref{main}) then becomes
\begin{equation}
\begin{aligned}
& \underset{\Pi(\mathcal{I})}{\mathrm{minimize}}
& &  \| g_{\epsilon}(s) - \hat{g}_{\epsilon}(s) \|_{\mathcal{H}_{2},\bar{\omega}}.
\end{aligned} \label{mainr}
\end{equation}
This minimization, however, is intractable given that its objective function is non-convex in $\Pi$, and that $\Pi$ itself is a combinatorial function of $\mathcal{I}$. To circumvent this problem, we borrow two relaxation steps from our recent paper \cite{nan} for solving (\ref{mainr}). The first relaxation (upper bound relaxation) is used to find an explicit function as the upper bound for the objective function in (\ref{mainr}), while the second relaxation (low-rank approximation) is used to simplify the computational complexity required in constructing the first relaxation. Unlike \cite{nan}, both relaxations here are posed in terms of the norm $\| \cdot \|_{\mathcal{H}_{2},\bar{\omega}}$ instead of the standard $\mathcal{H}_{2}$ norm to target the inter-area oscillation range.

\subsection{Upper Bound Relaxation}

After a few derivations based on (\ref{mainr}), the first relaxation step reduces to the optimization problem 
\begin{equation}
\begin{aligned}
& \underset{\Pi(\mathcal{I})}{\mathrm{minimize}}
& &  \xi = \| (I - \Pi^{T}\Pi)\Phi^{\frac{1}{2}}\|_{F},
\end{aligned}
\label{xi} \tag{RL-1}
\end{equation}
where $\Phi :=\Phi^{\frac{1}{2}} \Phi^{\frac{T}{2}} = \int_{-\bar{\omega}}^{\bar{\omega}} (\mathfrak{i}\omega - A_{\epsilon} + GX_{\epsilon})^{-1}B_{d}B_{d}^{T}(-\mathfrak{i}\omega - A_{\epsilon}^{T} + X_{\epsilon}G)^{-1} \mathrm{d}\omega \succ 0$ is the solution of the Lyapunov equation 
\begin{align}
(A_{\epsilon} - GX_{\epsilon})\Phi & + \Phi(A_{\epsilon} - GX_{\epsilon})^{T} \nonumber \\
& + S(\bar{\omega})B_{d}B_{d}^{T} + B_{d}B_{d}^{T}S(\bar{\omega}) = 0
\label{lyap}
\end{align}
with matrix $S(\bar{\omega})$ defined by
\begin{align}
& \ S(\bar{\omega}) = \frac{1}{2\pi} \int_{-\bar{\omega}}^{\bar{\omega}} (\mathfrak{i}\omega I - A_{\epsilon} + GX_{\epsilon})^{-1} \mathrm{d}\omega .
\end{align}
The basic methodology involved in this relaxation is that $\xi$ serves as an upper bound for the objective function of (\ref{mainr}). Therefore, by minimizing $\xi$ the matching error $\| g_{\epsilon}(s) - \hat{g}_{\epsilon}(s)\|_{\mathcal{H}_{2},\bar{\omega}}$ bounded below can be made small as well, which then helps in attaining the stability of $\hat{g}_{\epsilon}(s)$. Given that the derivation for (\ref{xi}) follows from \cite{nan} except for a few discrepancies in proofs due to the different norm metric. For a complete understanding of how (\ref{xi}) follows from (\ref{mainr}), and the associated stability condition for $\hat{g}_{\epsilon}(s)$, we refer the reader to Appendix B. 

Ideally speaking, (\ref{xi}) can be readily applied for designing $\mathcal{I}$, but its construction requires matrix $\Phi$, which is the solution of the Lyapunov equation (\ref{lyap}), and is subject to $\mathcal{O}(n^{3})$ computational complexity. This violates the computation constraint of (\ref{main}) as we want $\hat{K}$ to be numerically cheaper than $K$. To bypass this difficulty, we next describe an additional relaxation based on (\ref{xi}) that can be constructed in a simple way.

\subsection{Low-Rank Approximation}

The second relaxation is intuited by an explicit expression of the matrix $\Phi$ as follows.
\begin{lemma}
Denote the Hamiltonian matrix $H$ associated with ARE (\ref{areep}) and its  stable invariant subspace by 
\begin{align}
H = \begin{bmatrix}
A_{\epsilon} & - G \\
-Q & - A_{\epsilon}^{*}
\end{bmatrix}
,\quad H \begin{bmatrix}
Z \\ \bar{Z}
\end{bmatrix} = \begin{bmatrix}
Z \\ \bar{Z}
\end{bmatrix}\Lambda,
\label{hamieig}
\end{align}
where $\Lambda = diag([\lambda_{1},...,\lambda_{4n}])$ consists of all the eigenvalues of $H$ located in the left half plane, and denote the real and imaginary parts of the $i^{th}$ eigenvalue as $\lambda_{i} = a_{i} + \mathfrak{i}b_{i}$. Matrix $\Phi$ can be written as 
\begin{align}
\Phi = Z\mathcal{C}Z^{*}, \label{gramsolu}
\end{align}
where $\mathcal{C}$ is a Cauchy-like matrix with its entries defined by
\begin{align}
\mathcal{C}_{ij} & = -\frac{[Z^{-1}]_{i,:} B_{d}B_{d}^{*} [Z^{-*}]_{:,j}}{\lambda_{i} + \lambda_{j}^{*} } (c_{i} + c_{j}^{*}), \label{cauchyim} \\
c_{i} & = \frac{1}{2\pi} [\theta_{ci} - \mathfrak{i}\ln(\frac{a_{i}^{2}+(b_{i}-\bar{\omega})^{2}}{a_{i}^{2}+(b_{i}+\bar{\omega})^{2}})], \\
\theta_{ci} & = \arctan\bigg(\frac{b_{i}-\bar{\omega}}{a_{i}}\bigg) - \arctan\bigg(\frac{b_{i}+\bar{\omega}}{a_{i}}\bigg). \label{thetac}
\end{align}
\label{phiex}
\end{lemma}

The construction of $\Phi$ in (\ref{gramsolu}) requires full knowledge of $Z$ and $\Lambda$, which requires eigen-decomposition of $H$ that is subject to the $\mathcal{O}(n^{3})$ computational complexity. Moreover, since $H$ is large and asymmetric, its eigen-decomposition may not be well-defined due to the numerical difficulties \cite{golub}. One would, therefore, prefer to compute only the partial eigenspace and eigenvalues of $H$ using Krylov subspace-based techniques. Accordingly, we approximate $\Phi$ by a low-rank matrix $\Phi_{\kappa}$ defined as follows.

\begin{definition}
(I) Define $\Phi_{\kappa}\in \mathbb{R}^{4n\times 4n}$ as 
\begin{align}
\Phi_{\kappa} := \Phi_{\kappa}^{\frac{1}{2}}\Phi_{\kappa}^{\frac{T}{2}} =[ Z ]_{:,1:\kappa} [\mathcal{C}]_{1:\kappa,1:\kappa} [ Z^{*} ]_{1:\kappa,:}, \label{Phike}
\end{align} 
where $[ Z ]_{:,1:\kappa}$ and $[\mathcal{C}]_{1:\kappa,1:\kappa}$ are respectively the $\kappa$-dimensional submatrices from $Z$ and $\mathcal{C}$. Matrix $\Phi_{\kappa}^{\frac{1}{2}}$ can be found by $\Phi_{\kappa}^{\frac{1}{2}} = [ Z ]_{:,1:\kappa} [\mathcal{C}]_{1:\kappa,1:\kappa}^{\frac{1}{2}}$. 
\label{Phik}
\end{definition}
Using the expression in (\ref{Phike}), we replace $\Phi$ by $\Phi_{\kappa}$, and therefore, propose the relaxation for (\ref{xi}) as
\begin{equation}
\begin{aligned}
& \underset{\Pi(\mathcal{I})}{\mathrm{minimize}}
& &  \xi_{\kappa} = \| (I - \Pi^{T}\Pi) \Phi_{\kappa}^{\frac{1}{2}}\|_{F}.
\end{aligned}
\label{xik} \tag{RL-2}
\end{equation}
The optimality gap between these two optimizations is quantified by the following theorem.

\begin{theorem}
Assume that $Z^{-1}$ has a moderate condition number $\eta$, and each column of $B_{d}$ has a unitary norm. The optimum $\xi_{\kappa*}$ of (\ref{xik}) and the corresponding projection $\Pi_{*} = \mathrm{argmin}\ \xi_{\kappa}$ yield a worst-case error for (\ref{xi}) as 
\begin{align}
\| (I - \Pi_{*}^{T}\Pi_{*})\Phi^{\frac{1}{2}} \|_{F} - \xi_{\kappa*}  \leq  \underbrace{ \sqrt{\eta^{2} n_{d} \sum_{i=\kappa+1}^{4n} -\frac{\theta_{ci}}{2\pi a_{i}}} }_{e},
\label{xierror}
\end{align}
where $\lambda_{i} = a_{i} + \mathfrak{i}b_{i}$, and $\theta_{ci}$ is defined in (\ref{thetac}).
\label{kerror}
\end{theorem}

The preceding theorem shows that by solving (\ref{xik}) and applying the projection $\Pi_{*}$, the difference between the resulting value of $\xi$ and the optimum $\xi_{\kappa *}$ is bounded by the error term $e$. This implies that (\ref{xik}) will be most effective in approximating (\ref{xi}) if $e$ is kept small. Note that Definition \ref{Phik} only provides a constructing format for $\Phi_{\kappa}$, while the final expression for $\Phi_{\kappa}$ may vary with respect to different orders of eigenvalues in $\Lambda$. We next explain how to determine the ordering of the eigenvalues $\{ \lambda_{1},...,\lambda_{4n}\}$ to tighten the gap between (\ref{xi}) and (\ref{xik}).

\subsection{Constructing (\ref{xik}) for WAC}

For the power system model (\ref{psff}), the smallness of $e$ in (\ref{xierror}) follows naturally from the consideration of damping only the low-frequency inter-area oscillations. Recall that the frequency of inter-area oscillation modes is significantly smaller than that of the fast intra-area oscillation modes (more than $2$ Hz), and their damping factors are much smaller as well \cite{florian}. When the LQR matrices $Q$ and $R$ are chosen with moderate norms, $H$ will inherit this separation property from $A$, and will exhibit two spectral gaps for the real and imaginary parts of its eigenvalues as follows:
\begin{align}
0 > a_{1} \geq \cdots \geq a_{\kappa} \gg a_{\kappa+1} \geq \cdots \geq a_{4n},  \label{realgap}
\end{align}
\begin{align}
0 < |b_{1}| \leq \cdots \leq |b_{\kappa}| \ll |b_{\kappa+1}| \leq \cdots \leq |b_{4n}| . \label{imaggap}
\end{align}
The definition of $\kappa$ for WAC will be provided shortly. The two spectral gaps (\ref{realgap}) and (\ref{imaggap}) contribute to a small $e$ in two different ways.

{\bf 1. Damping factors:} Following (\ref{xierror}), the value of $e$ is proportional to $\theta_{ci}$, and is inversely proportional to $a_{i}$, $i=\kappa+1,...,4n$. Thus, the large magnitude of $|a_{i}|$, $i=\kappa+1,...,4n$ helps in attaining a small value of $e$. 

{\bf 2. Oscillation frequencies:} The spectral gap for the imaginary part, on the other hand, contributes to a small $\theta_{ci}$. The scalar $\theta_{ci}$ defined by (\ref{thetac}) represents the angular range of the perturbation $\pm \mathfrak{i}\bar{\omega}$ around $\lambda_{i}$, as illustrated in Fig. \ref{angc}. Recall that $[0,\bar{\omega}]$ for our design is limited to the inter-area frequency range only. Thus, $\bar{\omega}$ has similar magnitude as the low frequencies $b_{i}$, $i=1,...,\kappa$, and yields a moderate angular perturbation $\theta_{ci}$ as shown in Fig. \ref{angc1}. Due to the second spectral gap (\ref{imaggap}), $\bar{\omega}$ is significantly smaller than all the high frequencies $b_{i}$, $i=\kappa+1,...,4n$. As a result, the perturbation $\pm \bar{\omega}$ becomes almost negligible compared to $b_{i}$ for all $i=\kappa{+}1,...,4n$, which results in a sufficiently small $\theta_{ci}$ as shown in Fig. \ref{angc2}. 

Combining the two spectral gaps (\ref{realgap}) and (\ref{imaggap}), we complete the definition of $\Phi_{\kappa}$ as follows.
\begin{definition}
(II) Continuing from Definition \ref{Phik}, the definition of index $\kappa$ for the wide-area control problem is such that
$$ 0 < |\lambda_{1}| \leq \cdots \leq |\lambda_{\kappa}| \ll |\lambda_{\kappa+1}| \leq \cdots \leq |\lambda_{4n}| .$$
\label{phiex2}
\end{definition}

In this order, eigenvalues from $\lambda_{\kappa+1}$ through $\lambda_{4n}$ have larger magnitudes for both real and imaginary parts compared to the other eigenvalues. By Definitions \ref{Phik} and \ref{phiex2}, one, therefore, only needs to compute the first $\kappa$ eigenvalues of $H$ with smallest magnitudes from $H$, and their $\kappa$ eigenvectors denoted by $[ Z ]_{:,1:\kappa}$. The first $\kappa$ rows of $Z^{-1}$ can be approximated by the pseudo-inverse of $[ Z ]_{:,1:\kappa}$. These $\kappa$ smallest eigenvalues and eigenvectors can be efficiently solved by Arnoldi algorithm in $\mathcal{O}(n\kappa^{2})$ time \cite{golub}. Therefore, if $\kappa\ll 4n$, the construction of (\ref{xik}) can be significantly simpler than $\mathcal{O}(n^{3})$ required for (\ref{xi}). Note that although the spectral separations (\ref{realgap}) and (\ref{imaggap}) help in a close matching between the two optimizations (\ref{xi}) and (\ref{xik}), it is, however, not necessary for the gap $|\lambda_{\kappa}|\ll |\lambda_{\kappa+1}|$ to exist in order to apply the relaxation (\ref{xik}). Therefore, even if the power system model (\ref{psff}) does not have a significant spectral separation, one can still apply the low-rank approximation $\Phi_{\kappa}$. The upshot will be that the optimality gap specified by (\ref{xierror}) will increase in that case, but the computation of $\Phi_{\kappa}$ will still remain more scalable than that of $\Phi$. 

\begin{figure}
    \centering
    \begin{subfigure}[t]{0.46\columnwidth}
    \centering
        \includegraphics[width=0.95\columnwidth]{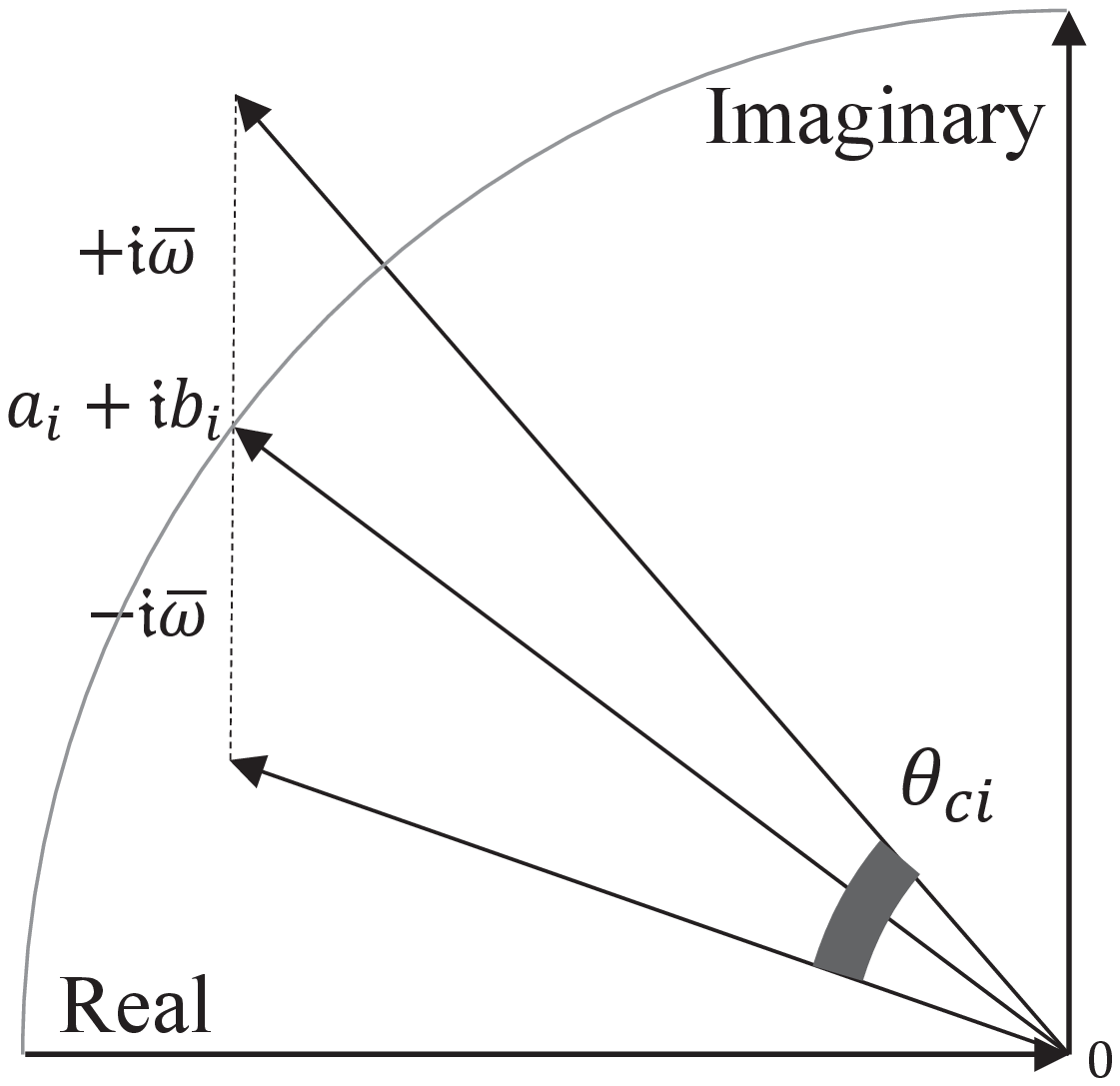}
        \caption{$\bar{\omega}$ in similar magnitude as $b_{i}$}
        \label{angc1}
    \end{subfigure}
    ~        
    \begin{subfigure}[t]{0.46\columnwidth}
    \centering
        \includegraphics[width=0.95\columnwidth]{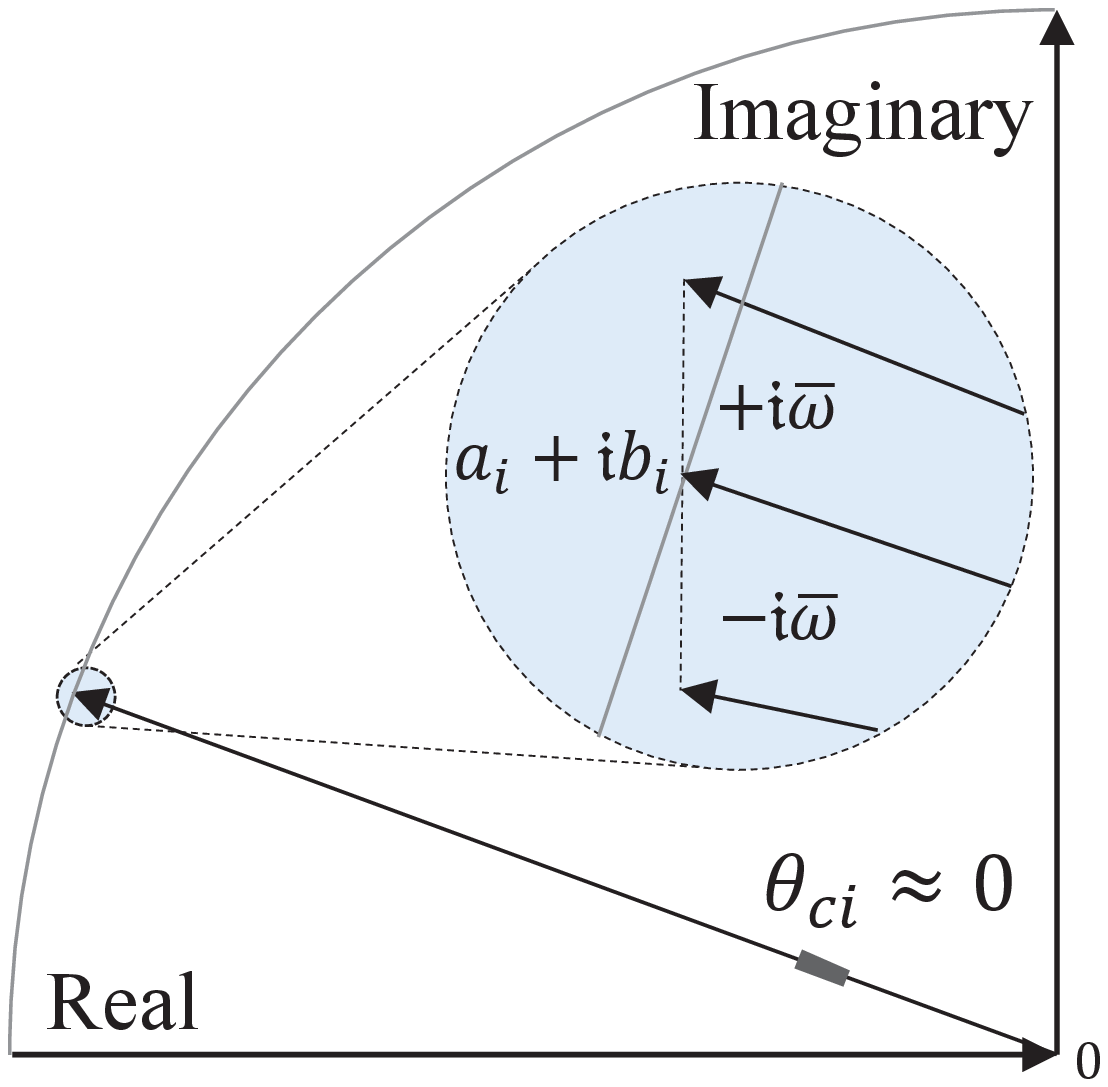}
        \caption{$\bar{\omega} \ll | b_{i} |$}
        \label{angc2}
    \end{subfigure}
    \caption{Interpretation of the angle $\theta_{ci}$}
    \label{angc}
\end{figure}

\subsection{Design of $\mathcal{I}$}

To illustrate the final design of the generator clustering set $\mathcal{I}$, we reduce the problem (\ref{xik}) into its minimal form with respect to $\mathcal{I}$. Recall that in (\ref{Pstack}) the projection matrix $\Pi$ is defined over a block-diagonal structure $\Pi = I_{4}\otimes P$ to preserve the generator identities. By removing this redundancy, the objective function $\xi_{\kappa}$ in (\ref{xik}) can be rewritten in terms of $P$ as 
\begin{align}
\xi_{\kappa} = \| (I - P^{T}P) W\Psi \|_{F},
\end{align}
where $W=diag(w)$ follows from the same weight $w$ specified in Section \RNum{4}, and matrix $\Psi$ is defined by
\begin{align}
\Psi & = W^{{-}1}\big[ [\Phi_{\kappa}^{\frac{1}{2}}]_{1:n,:}\ [\Phi_{\kappa}^{\frac{1}{2}}]_{n{+}1:2n,:} \ ... \nonumber \\
& \quad \quad \quad \quad \quad \quad \ \ ...\ [\Phi_{\kappa}^{\frac{1}{2}}]_{2n{+}1:3n,:}\ [\Phi_{\kappa}^{\frac{1}{2}}]_{3n{+}1:4n,:} \big].
\end{align}
We further denote the row vectors in $\Psi$ by 
\begin{align}
\Psi = \begin{bmatrix}
\psi_{1} & \cdots & \psi_{n}
\end{bmatrix}^{T}.
\end{align}
In these notations, (\ref{xik}) can be rewritten as
\begin{equation}
\begin{aligned}
& \underset{\mathcal{I}_{1},...,\mathcal{I}_{r}}{\mathrm{minimize}}
& &  \xi_{\kappa}^{2} = \sum_{j=1}^{n} w_{j}^{2}\| \psi_{j} - c_{i}\|_{2}^{2} \\
& \quad \ \mathrm{s.t.} & & c_{i} =  \frac{\sum_{j \in \mathcal{I}_{i}} w^{2}_{j}\psi_{j}}{\sum_{j \in \mathcal{I}_{i}} w^{2}_{j}}
\end{aligned}\quad .
\label{wkm}
\end{equation}
This optimization problem is in the same form as standard k-means clustering, where the problem is to assign the vectors $\psi_{j}$ among $r$ clusters such that the vectors $\psi_{j}$ inside each cluster are close to each other in the sense of their weighted distances. If the number of clusters $r$ is fixed, the optimal solution of (\ref{wkm}) can be found in exact $\mathcal{O}(n^{4r\kappa +1})$ time. In practice, however, problem (\ref{wkm}) is usually approached by heuristic algorithms that can provide good local optimum under reasonable numerical complexity. For the sake of this paper, we apply the simplest algorithm called Lloyd's algorithm \cite{kmeans} for solving (\ref{wkm}), which requires $\mathcal{O}(n\kappa r k)$ complexity, where $k$ represents the number of iterations. The design of $\hat{K}$ using Lloyd's algorithm is summarized in Alg. \ref{algk}.

\begin{algorithm}
    \SetKwInOut{Input}{Input}
    \SetKwInOut{Output}{Output}
    \Input{$(A,B,Q,R,B_{d},M)$, $\kappa$ and $r$}
 Find clustering weight $w=\bar{v}$ from (\ref{vbar})\;
 Construct $\Phi_{\kappa}^{\frac{1}{2}}$ according to Definition \ref{Phik} and \ref{phiex2}\;
 {\bf{Initialization (Lloyd's algorithm) }:} Assign $r$ random rows from $\Psi = [\psi_{1},...,\psi_{n}]^{T}$ as the initial centroids $c_{1}^{0},...,c_{r}^{0}$\;
 Find initial clustering sets $\mathcal{I}^{0} = \{ j \to \mathcal{I}_{i}^{0} \ | \ \underset{i=1,...,r} {\mathrm{argmin}}\ w^{2}_{j} \|\psi_{j}-c_{i}^{0}\|_{2}^{2},\ j=1,...,n \}$\;
 Update the centroids: $c_{i}^{0}= \frac{\sum_{j \in \mathcal{I}_{i}^{0}} w^{2}_{j}\psi_{j}}{\sum_{j \in \mathcal{I}_{i}^{0}} w^{2}_{j}}$, $i=1,...,r$ \;
 $k=1$\;
 \While{$\mathcal{I}^{k-1}\neq \mathcal{I}^{k}$ or within maximum iterations}{
  Update clustering sets $\mathcal{I}^{k} = \{ j \to \mathcal{I}_{i}^{k} \ | \ \underset{i=1,...,r} {\mathrm{argmin}}\ w^{2}_{j}\|\psi_{j}-c_{i}^{k-1}\|_{2}^{2},\ j=1,...,n \}$\;
  Update the centroids: $c_{i}^{k}= \frac{\sum_{j \in \mathcal{I}_{i}^{k}} w^{2}_{j}\psi_{j}}{\sum_{j \in \mathcal{I}_{i}^{k}} w^{2}_{j}}$, $i=1,...,r$ \;
  $k=k+1$ \;
 }
 Solve $\hat{K}$ from control inversion steps (\ref{Pdefe}-\ref{Khat}) with $\mathcal{I} = \mathcal{I}^{k}$ and $w=\bar{v}$\;
 \Output{$\hat{K}$}
 \caption{Overall Design of wide-area controller $\hat{K}$ } \label{algk}
\end{algorithm}

\subsection{Computational Complexity}

We close this section by summarizing the overall computational complexity for designing $\hat{K}$ using control inversion. As in Alg. \ref{algk} one starts with constructing the relaxation (\ref{xik}), which is subject to $\mathcal{O}(n\kappa^{2})$, and then solves $\mathcal{I}$ using Lloyd's algorithm in $\mathcal{O}(n\kappa r k)$ time. With the resulting $\mathcal{I}$ and $w$ analytically found in Theorem \ref{wsolution}, $\hat{K}$ can be constructed using control inversion steps (\ref{Pdefe}-\ref{Khat}) in $\mathcal{O}(n^{2}r)$ time. Therefore, the overall complexity for designing $\hat{K}$ is $\mathcal{O}(n\kappa^{2}) + \mathcal{O}(n\kappa r k) + \mathcal{O}(n^{2}r)$, or predominantly $\mathcal{O}(n^{2}r)$ if $\kappa \leq r$. If $\kappa$ and $r$ are much smaller than $n$, this complexity will be far simpler than $\mathcal{O}(n^{3})$ of optimal LQR. It is worth mentioning that the computational saving of this design is mainly facilitated by the low-rank approximation (\ref{xik}). However, even without this approximation, the overall complexity, although $\mathcal{O}(n^{3})$ which is same as optimal LQR, is still more scalable than the designs posed in \cite{florian,wu} that rely on semi-definite programming subject to $\mathcal{O}(n^{4})$ complexity.

\section{Case Study}

In this section, we verify our proposed design using the NPCC $48$-machine model. The model represents the region of Northeastern Power Coordinating Council (NPCC) with the geography and locations of all $48$ machines shown in Fig. \ref{npcc}. The parameters for all the synchronous machines, transmission lines and loads for this model are provided in PST toolbox \cite{pst}. Using these parameters\footnote{Note that in \cite{pst} machines $\{ 15,23{:}27,33{:}35,37{:}48\}$ are not provided with the excitation time constant $T_{doi}^{\prime}$. We choose $T_{doi}^{\prime}$ for these machines to lie between $4$ s and $6$ s, which are comparable to the time constants of the other generators.}, we determine the operating point of the system by solving power flow, and then construct the linearized network model (\ref{psff}). The open-loop model so constructed exhibits $4$ slow oscillation modes, and their frequencies are all less than $2$ rad/s. Therefore, we set $\bar{\omega}=2$ in (\ref{main}) for evaluating the closed-loop performance of $\hat{K}$, and we set $\kappa = 4$ for constructing the low-rank approximation (\ref{xik}). The reference LQR controller $K$ is defined by $R = I$ and
$$Q = (I_{4}{\otimes} M^{\frac{1}{2}})^{{-}1} diag(I_{n} {-} \mathbf{1}_{n}\mathbf{1}_{n}^{T}/n, I_{n}, I_{n}, I_{n})(I_{4}{\otimes} M^{\frac{1}{2}})^{{-}1}.$$
This choice of $Q$ penalizes the oscillations in the generator angle differences, and also satisfies the CPLQR condition in Lemma \ref{dclqr}. We also assume that the disturbance enters the system dynamics through machines $\{ 27{-}30 \}$ as shown in Fig. \ref{npcc}, which means that $B_{d}$ equals the $\{ 27{-}30 \}^{th}$ columns of $B$. The disturbance $d$ is treated as a unit impulse to mimic the effect of a fault. For the performance output $y=Cx$, we let $\bar{C} = [\mathbf{1}_{n-1}\ -I_{n-1}]$ in (\ref{perfc}) to evaluate the angle differences between generator $1$ and all the remaining generators.

\begin{figure}
    \centering
        \includegraphics[width=0.9\columnwidth]{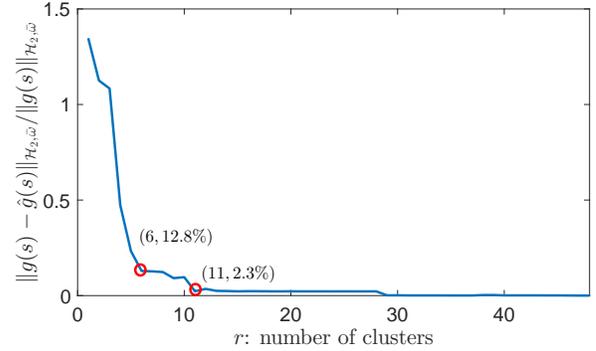}
        \caption{Performance matching between $K$ and $\hat{K}$}
        \label{err}
\end{figure}
\begin{figure*}
    \centering
    \begin{subfigure}[t]{0.3\textwidth} 
    \centering
        \includegraphics[width=0.925\columnwidth]{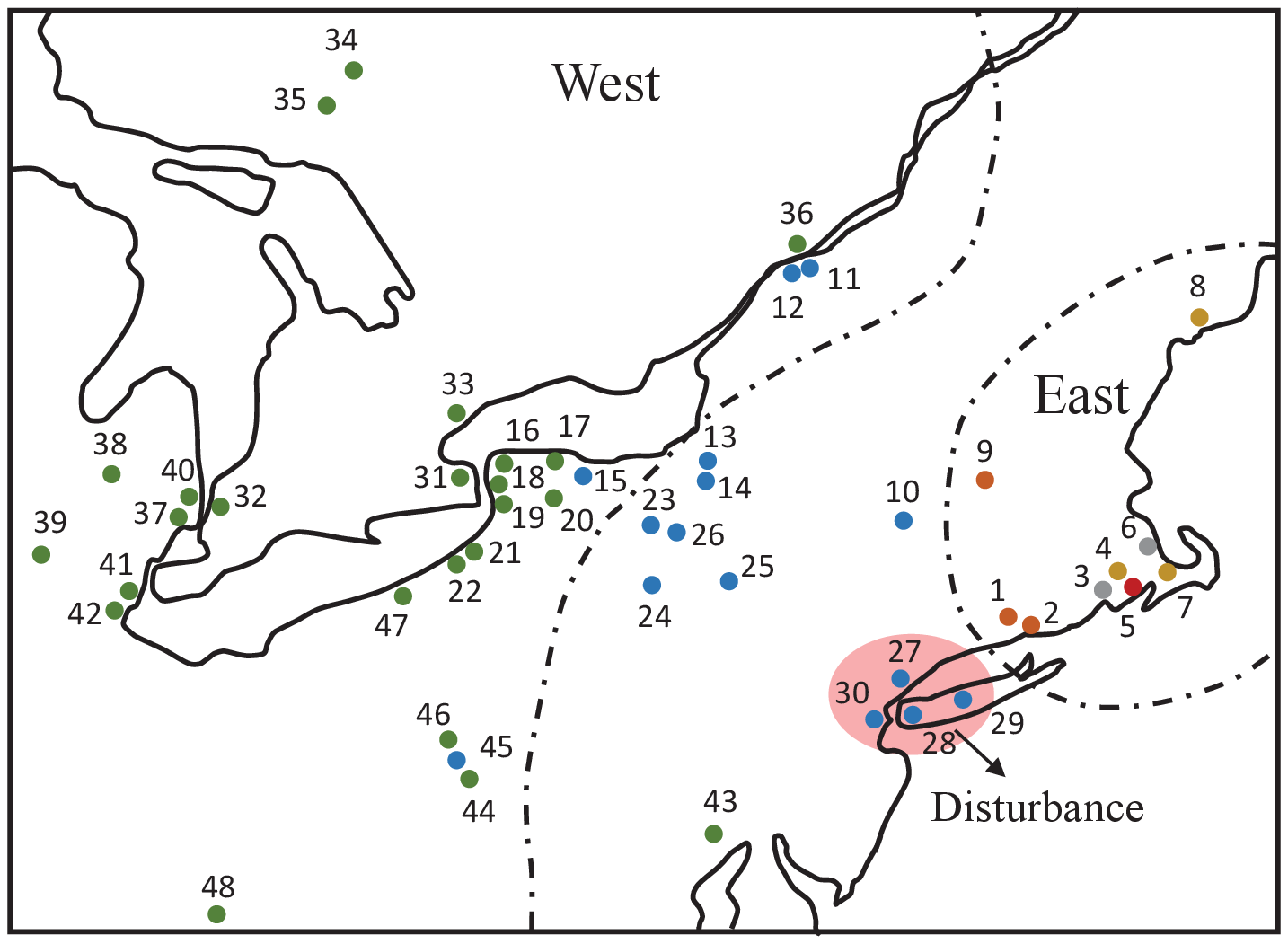}
        \caption{$r=6$}
        \label{npcc1}
    \end{subfigure}
    ~        
    \begin{subfigure}[t]{0.3\textwidth}
    \centering
        \includegraphics[width=0.925\columnwidth]{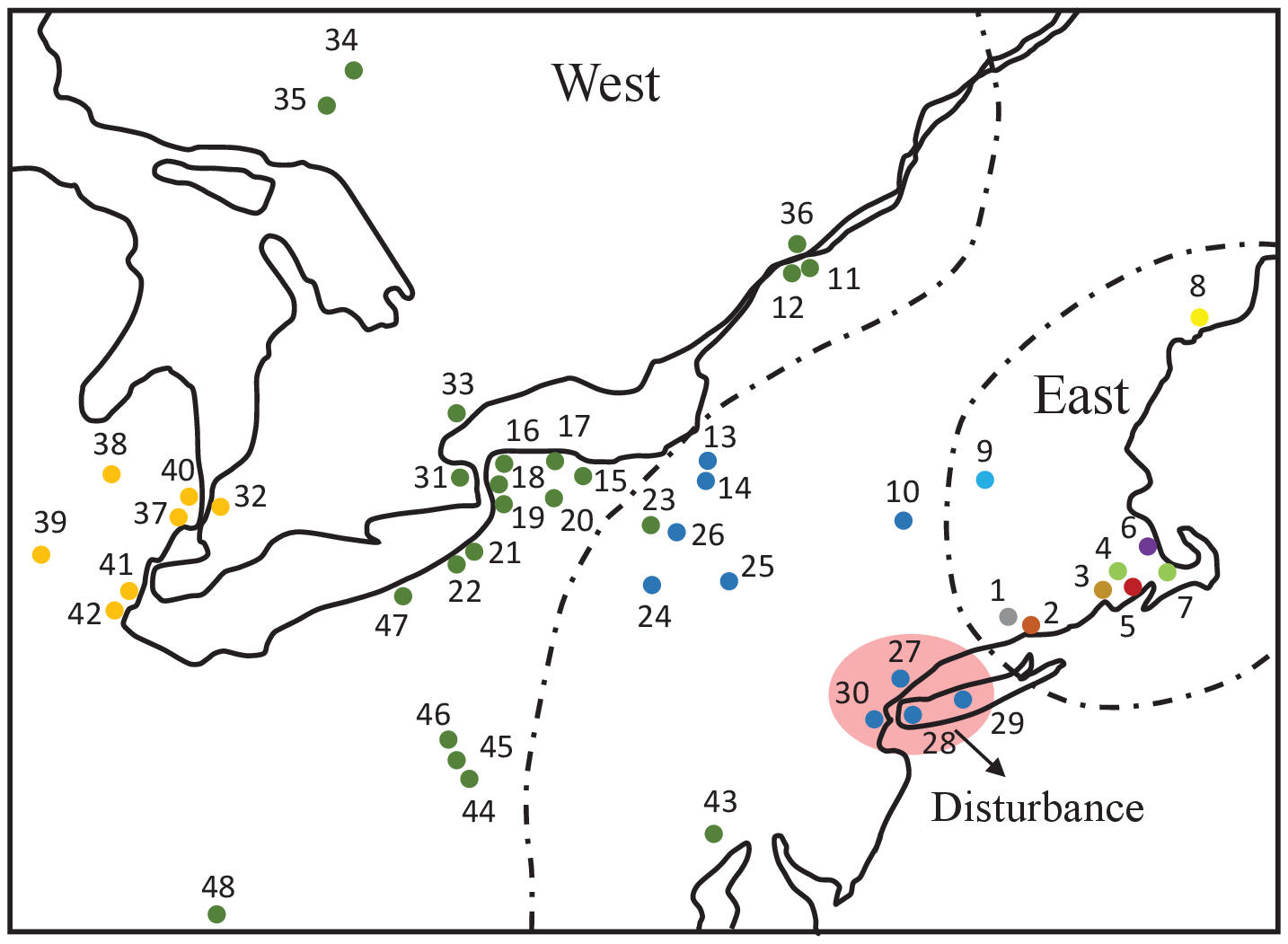}
        \caption{$r=11$ }
        \label{npcc2}
    \end{subfigure}
    ~        
    \begin{subfigure}[t]{0.3\textwidth}
    \centering
        \includegraphics[width=0.925\columnwidth]{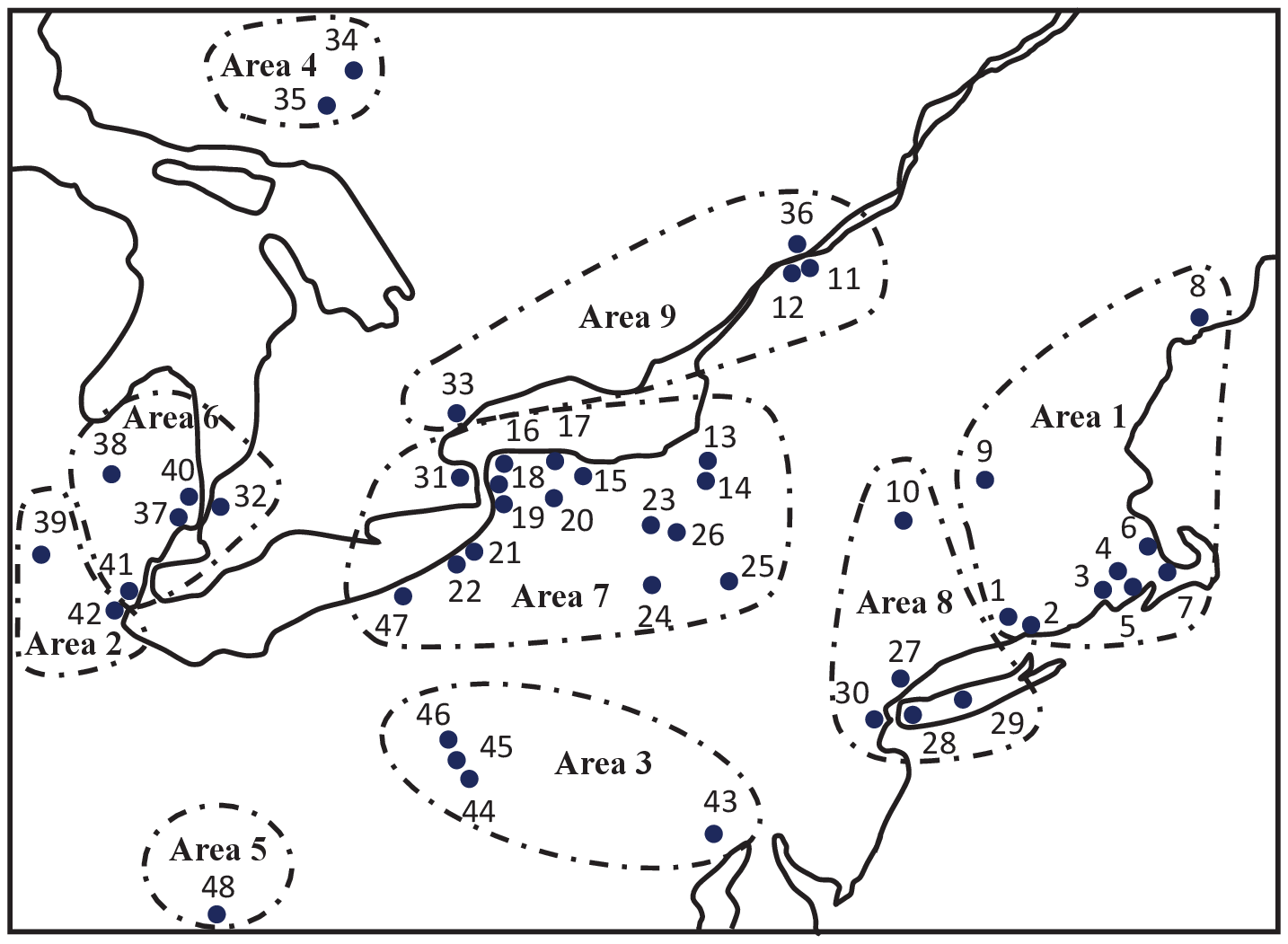}
        \caption{$9$-area partitions from coherency}
        \label{npccco}
    \end{subfigure}
    \caption{Comparison between coherent groups and clustering assignment of $\hat{K}$}\label{npcc}
\vspace{-1em}
\end{figure*}

\subsection{Wide-Area Control of NPCC}

We determine the clustering set $\mathcal{I}$ from (\ref{xik}), and then design the controller $\hat{K}$ with the number of clusters $r$ varying from $1$ to $48$. The resulting controllers are compared with the optimal controller $K$ in Fig. \ref{err}, where the performance metric is the objective function of (\ref{main}) normalized by $\| g(s) \|_{\mathcal{H}_{2},\bar{\omega}}$. It is worth mentioning that solving $\mathcal{I}$ based on (\ref{xi}) only (i.e., without any low-rank approximation), yields the same results as in Fig. \ref{err}. This verifies the effectiveness of (\ref{xik}) in matching (\ref{xi}). From Fig. \ref{err}, the matching error decreases to zero when $r$ scales up to $n=48$. Recall that for simplicity of design and implementation it is preferable to keep $r$ small while maintaining a relatively close performance matching. Two cases that achieve both of these conditions are for $r=6$ and $r=11$, yielding $12.8\%$ and $2.3\%$ matching errors respectively. In terms of the structure of $\hat{K}$, we illustrate the resulting clustering assignments for $r=6$ and $r=11$ in Fig. \ref{npcc}, where machines marked by the same color are assigned to the same cluster. As shown in Fig. \ref{npcc}, the cluster assignments for both $r=6$ and $r=11$ closely resemble the geographical partitions of the actual NPCC system. The distinction is that when $r$ changes from $6$ to $11$, the machines in the western region form one additional cluster, while the machines in the east split up into multiple clusters. These newly formed clusters are geographically close to or are contained inside the clusters corresponding to $r=6$. This means the implementation architecture shown in Fig. \ref{arch} for $r=6$ can also be applied for $r=11$ since it may be possible for the VMs to multi-task the implementation steps required by these extra clusters in their geographical region. Thus, one can choose $r=11$ as the best choice for $r$ in this case, achieving a $2.3\%$ matching error while still maintaining a simple implementation structure as required by $r=6$.

We also compare these two clustering assignments with the $9$ coherency-based generator clusters of the NPCC system as shown in Fig. \ref{npccco}. These coherent areas are partitioned based on the spectral characteristics of the open-loop state matrix. Depending on the power system model, they may represent operating regions of different utility companies. As is obvious from Fig. \ref{npcc}, the cluster assignments for our design for both $r=6$ and $r=11$ differ from the coherent groups, indicating the generators across different utility areas may need to be clustered together for taking the wide-area control action. This observation pinpoints to the fact that WAC should not be limited to coherency-based partitioning.

\subsection{Numerical Savings}

The computation time required for solving the optimal controller $K$ is $0.65$ second in a standard computer, and that for $\hat{K}$ with $r=11$ is only $0.16$ seconds. This computational saving may seem insignificant as the dimension of the NPCC model ($4n{=}192$) is still small compared to realistic power systems where $n$ can be in thousands. To verify the scalability of our design for such larger systems, we compare the computational costs between $K$ and $\hat{K}$ using models with the number of generators ranging from $100$ to $1000$ as shown in Fig. \ref{scala}. The controller $\hat{K}$ in these cases are all designed with $r=11$ and $\kappa =4$. These test models were generated from (\ref{psff}) using randomized but realistic admittance matrices, generator parameters and operating points. As is clear from Fig. \ref{scala}, the design of $\hat{K}$ becomes significantly more scalable than that of the optimal controller $K$ when the dimension of the power system grows. For example, at $n=1000$ the computation time for $\hat{K}$ is only $39.7$ seconds in total, while it requires $568.3$ seconds to solve $K$. This verifies the $\mathcal{O}(n^{2}r)$ complexity of Algorithm \ref{algk}. The best use of our proposed method, therefore, is for very large values of $n$.

\begin{figure}
    \centering
        \includegraphics[width=0.95\columnwidth]{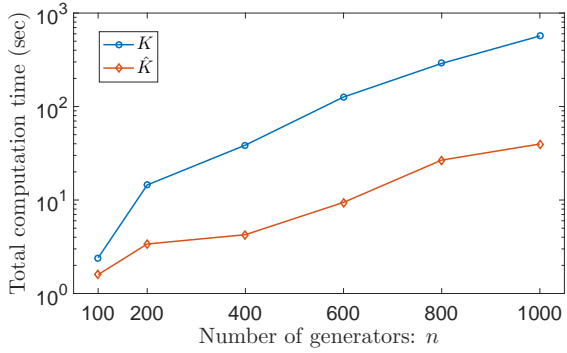}
        \caption{Scalability results for $\hat{K}$ with $r=11$ and $\kappa=4$ }
        \label{scala}
\end{figure}

\section{Conclusion}
In this paper we developed a structured suboptimal LQR controller for wide-area oscillation damping control of large power systems. The control design is approached by a control inversion strategy, which results in a simpler lower-dimensional design and a hierarchical implementation. We compared the numerical efficiency of this method with standard LQR, and also showed how the spectral characteristics of the open-loop model can enhance this efficiency. Our future work will be to extend the control inversion concept from a model-based approach to a model-free approach.

\section*{Appendix A. State Matrices in Model (\ref{psf})}
The linearized matrices in (\ref{psf}) are given as follows.
\begin{align*}
[L_{1}]_{ii} & = - \sum_{j=1,j\neq i}^{n} [L_{1}]_{ij}, \\
[L_{1}]_{ij} & = - E_{qi0}^{\prime}E_{qj0}^{\prime}Y_{\alpha,ij}\cos(\delta_{i0} - \delta_{j0} - \alpha_{ij}), \\
[L_{2}]_{ii} & = - \sum_{j=1,j\neq i}^{m} [L_{2}]_{ij}, \\
[L_{2}]_{ij} & = - \frac{x_{di} - x_{di}^{\prime}}{x_{di}^{\prime}} E_{qj0}^{\prime}Y_{\beta,ij}\sin(\delta_{i0} - \delta_{j0} - \beta_{ij}), \\
[L_{3}]_{ii} & = - \sum_{j=1,j\neq i} [L_{3}]_{ij}, \\
[L_{3}]_{ij} & = - \frac{V_{R}\sin(\delta_{i0} - \delta_{j0} - \beta_{ij}) - V_{I}\cos(\delta_{i0} - \delta_{j0} - \alpha_{ij})}{\sqrt{V_{R}^{2} + V_{I}^{2}}} \\
& \quad \cdot K_{Ai}Y_{\beta,ij}E_{qj0}^{\prime}, \\
[P_{1}]_{ii} & = - E_{qi0}^{\prime}Y_{\alpha,ii}\cos(\alpha_{ii}) \\
& \quad - \sum_{j=1}^{n}\ E_{qj0}^{\prime}Y_{\alpha,ij}\cos(\delta_{i0} - \delta_{j0} - \alpha_{ij}), \\
[P_{1}]_{ij} & = - E_{qi0}^{\prime}Y_{\alpha, ij}\cos(\delta_{i0} - \delta_{j0} - \alpha_{ij}), \\
[P_{2}]_{ii} & = - \frac{x_{di}}{x_{di}^{\prime}} - \frac{x_{di} - x_{di}^{\prime}}{x_{di}^{\prime}} Y_{\beta,ii}\cos(\beta_{ii}), \\
[P_{2}]_{ij} & = - \frac{x_{di} - x_{di}^{\prime}}{x_{di}^{\prime}} Y_{\beta,ij}\cos(\delta_{i0} - \delta_{j0} - \beta_{ij}), \\
[P_{3}]_{ii} & = - K_{Ai}Y_{\beta,ii} \frac{V_{R}\cos(\beta_{ii}) - V_{I}\sin(\beta_{ii})}{\sqrt{V_{R}^{2} + V_{I}^{2} }}, \\
[P_{3}]_{ij} & = - \frac{V_{R}\cos(\delta_{i0} - \delta_{j0} - \beta_{ij}) + V_{I}\sin(\delta_{i0} - \delta_{j0} - \beta_{ij})}{\sqrt{V_{R}^{2} + V_{I}^{2} }} \\
& \quad \cdot K_{Ai}Y_{\beta,ij}, \\
V_{R} & = \sum_{j=1}^{n} Y_{\beta,ij}E_{qj0}^{\prime}\cos(\delta_{i0} - \delta_{j0} - \beta_{ij}), \\
V_{I} & = \sum_{j=1}^{n} Y_{\beta,ij}E_{qj0}^{\prime}\sin(\delta_{i0} - \delta_{j0} - \beta_{ij}).
\end{align*}
Let $[Y]_{ij} = \tilde{Y}_{ij}e^{\mathfrak{i}\tilde{\theta}_{ij}}$ denote the admittance between the $i^{th}$ and $j^{th}$ buses, $i,j{=}1,...,n+n_{l}$, including the load-side impedances. Using matrix $Y$, the parameters $(\tilde{Y}_{\alpha,ij},\tilde{\alpha}_{ij},\tilde{Y}_{\beta,ij},\tilde{\beta}_{ij})$ follows from two equivalent matrices $[Y_{\alpha}]_{ij} = \tilde{Y}_{\alpha,ij}e^{\mathfrak{i}\tilde{\alpha}_{ij}}$ and $[Y_{\beta}]_{ij} = \tilde{Y}_{\beta,ij}e^{\mathfrak{i}\tilde{\beta}_{ij}}$ constructed by
\begin{align*}
Y_{\alpha} = Y_{11} - Y_{12}Y_{22}^{-1}Y_{21}, \quad Y_{\beta} = [Y_{\alpha} + diag(\mathfrak{i}x_{di}^{\prime})]^{-1}, 
\end{align*}
where $Y_{11}$, $Y_{12}$, $Y_{21}$ and $Y_{22}$ are submatrices of $Y$ partitioned according to the bus indices, with $Y_{11}$ corresponds to all the first $n$ generator buses. 

\section*{Appendix B. Derivation of (\ref{xi})}

The first relaxation (\ref{xi}) can then be derived from an upper bound on $\| g_{\epsilon}(s) - \hat{g}_{\epsilon}(s) \|_{\mathcal{H}_{2},\bar{\omega}}$ as follows. 
\begin{customlem}{B.1}
Denote the error system by $g_{e}(s) = g_{\epsilon}(s) - \hat{g}_{\epsilon}(s)$. The objective function in (\ref{mainr}) satisfies the inequality
\begin{align}
\| g_{\epsilon}(s) - \hat{g}_{\epsilon}(s) \|_{\mathcal{H}_{2},\bar{\omega}} \leq \gamma \| E\Phi^{\frac{1}{2}}\|_{F},
\label{bieq} 
\end{align}
where $E=X_{\epsilon}-\hat{X}$, and the scalar $\gamma$ in (\ref{bieq}) is any positive real number such that a real-valued matrix $\Gamma = \Gamma^{T} \succeq 0$ exists and satisfies 
\begin{align}
\begin{bmatrix}
\Gamma (A_{\epsilon}-G\hat{X})+(A_{\epsilon}-G\hat{X})^{T}\Gamma & \Gamma G & C^{T} \\
G\Gamma & -\gamma I & 0 \\
C & 0 & -\gamma I
\end{bmatrix} \prec 0,
\label{brl} 
\end{align}
and $\gamma$ is bounded if $\hat{g}_{\epsilon}(s)$ is stable. 
\label{bound1}
\end{customlem}
\begin{proof}
Using a coordinate transformation of $T=\begin{bmatrix} I & I \\ I & 0 \end{bmatrix}$ and $T^{-1}=\begin{bmatrix} 0 & I \\ I & -I \end{bmatrix}$, The error system $g_{e}(s)$ can be written as
\begin{align*}
g_{e}(s) = -C(sI - A_{\epsilon} + G\hat{X})^{-1}GE(sI - A_{\epsilon} + GX_{\epsilon})^{-1}B_{d}.
\end{align*}
The upper bound (\ref{bieq}) then follows directly from the triangle inequality of the norm $\| \cdot \|_{\mathcal{H}_{2},\bar{\omega}}$, which yields
\begin{align*}
\| g_{e}(s)\|_{\mathcal{H}_{2},\bar{\omega}} \leq &  \| C(sI - A_{\epsilon} + G\hat{X})^{-1}G\|_{\mathcal{H}_{\infty},\bar{\omega}} \\
& \quad \quad \cdot \| E(sI - A_{\epsilon} + GX_{\epsilon})^{-1}B_{d} \|_{\mathcal{H}_{2},\bar{\omega}}.
\end{align*}
The norm $\| \cdot \|_{\mathcal{H}_{\infty},\bar{\omega}}$ follows the similar definition of (\ref{h2w}) as
\begin{align}
& \|h(s)\|_{\mathcal{H}_{\infty},\bar{\omega}} = \underset{-\bar{\omega}\leq \omega \leq \bar{\omega}}{\mathrm{sup}}\ \bar{\sigma}[ g(j\omega)].
\end{align}
for any stable transfer function matrix $h(s)$. This norm represents the $\mathcal{H}_{\infty}$ norm of $\hat{g}_{\epsilon}(s)$ over a limited frequency range, and hence is bounded by the standard $\mathcal{H}_{\infty}$ norm over infinite frequency range. If $\hat{g}_{\epsilon}$ is stable, the value $\gamma$ then exists and upper bounds the $\mathcal{H}_{\infty}$ norm by (\ref{brl}) according to bounded real lemma \cite{rc}. The second norm on the RHS can be written directly as $\| E\Phi^{\frac{1}{2}}\|_{F}$ according to \cite{antoulas}.
\end{proof}

The next lemma provides a stability condition for $\hat{g}_{\epsilon}(s)$.
\begin{customlem}{B.2}
The system $\hat{g}_{\epsilon}(s)$ is asymptotically stable if
\begin{align}
\bar{\sigma}(E\Phi^{\frac{1}{2}})\bar{\sigma}(G)\bar{\sigma}(\Phi^{\frac{1}{2}}) < \underline{\sigma}(B_{d}B_{d}^{T}). \label{stabcond1}
\end{align}
\label{stabcond}
\end{customlem}
\begin{proof}
Consider a quadratic function $V(x) = x^{T}\Phi^{-1}x >0$, where $\Phi$ is the solution of (\ref{lyap}). For $\hat{g}_{\epsilon}(s)$ to be stable, $\dot{V}(x)$ needs to be negative, or equivalently
\begin{align}
(A_{\epsilon}-G\hat{X})^{T}\Phi^{-1} + \Phi^{-1} (A_{\epsilon} - G\hat{X}) \prec 0. \label{lyapine2}
\end{align} 
By pre- and post-multiplying (\ref{lyapine2}) with $\Phi$, and after a few calculations, (\ref{lyapine2}) yields
\begin{align*}
\Phi EG + GE\Phi \prec  B_{d}B_{d}^{T}.
\end{align*}  
The inequality above will be satisfied if
\begin{align*}
\bar{\lambda}(\Phi EG + GE\Phi) < \underline{\lambda}(B_{d}B_{d}^{T}).
\end{align*}
which is further satisfied if
\begin{align*}
\bar{\lambda}(\Phi EG {+} GE\Phi) {\leq} \bar{\sigma}(E\Phi^{\frac{1}{2}})\bar{\sigma}(G)\bar{\sigma}(\Phi^{\frac{T}{2}}) {<} \underline{\sigma}(B_{d}B_{d}^{T}).
\end{align*}
The inequality above proves the condition (\ref{stabcond1}).
\end{proof}

From Lemma \ref{bound1}, the objective function of (\ref{mainr}) is linearly depended on $\| E\Phi^{\frac{1}{2}}\|_{F}$ with respect to the scalar $\gamma$. Hence, one can simply approach the problem (\ref{mainr}) by minimizing the norm $\| E\Phi^{\frac{1}{2}}\|_{F}$. On the other hand, given that $\bar{\sigma}(E\Phi^{\frac{1}{2}}) \leq \| E\Phi^{\frac{1}{2}}\|_{F}$, minimization of $\| E\Phi^{\frac{1}{2}}\|_{F}$ also helps in attaining the inequality (\ref{stabcond1}) from Lemma \ref{stabcond}, which then guarantees the stability of $\hat{g}_{\epsilon}(s)$ and boundedness of $\gamma$. In the literature of model and controller reduction, this type of bound minimization has been commonly attempted (see \cite{antoulas} and the references therein), but under the assumption that $\Pi$ is unstructured. By this assumption, $E$ can be found as an explicit function of $\Pi$. In our case, however, $\Pi$ has a structure as in (\ref{Pdefe}) and (\ref{Pstack}), due to which this explicit functional relationship does not hold anymore. The next theorem, therefore, addresses this problem by finding an upper bound on $\| E\Phi^{\frac{1}{2}}\|_{F}$ as an explicit function of $\Pi$.
\begin{customthm}{B.3}
Denote $\xi = \| (I - \Pi^{T}\Pi)\Phi^{\frac{1}{2}} \|_{F}$. The norm $\| E\Phi^{\frac{1}{2}}\|_{F}$ satisfies the inequality
\begin{align}
& \| E\Phi^{\frac{1}{2}}\|_{F} \leq f(\xi) = \epsilon_{1} \bar{\sigma}(Q) \xi^{2} + 2 \epsilon_{1} \epsilon_{2} \xi ,
\label{rboundn}
\end{align}
where $\epsilon_{1} = \frac{\bar{\sigma}(\Phi^{-\frac{1}{2}}) }{\underline{\sigma}[\Phi^{-\frac{1}{2}}(A_{\epsilon}-GX_{\epsilon})\Phi^{\frac{1}{2}}]}$, $\epsilon_{2} = \tilde{\beta} \bar{\sigma}(A_{\epsilon}) \bar{\sigma}(\Phi^{\frac{1}{2}}) + \bar{\sigma}(Q\Phi^{\frac{1}{2}})$, $\epsilon_{3} = (\tilde{\beta}^{2} + 1 )\bar{\sigma}(\Phi) $, and $\tilde{\beta} \geq \mathrm{sup}_{\Pi}\ \bar{\sigma}(\tilde{X})$ are positive scalars that are independent of $\Pi$.
\label{tmain}
\end{customthm}
\begin{proof}
We derive an analytical expression for $E$ by projecting the lower-dimensional ARE (\ref{arerr}) to the full dimension as
\begin{align}
\Pi^{T}(\tilde{A}_{\epsilon}^{T}\tilde{X} + \tilde{X}\tilde{A}_{\epsilon} + \tilde{Q} - \tilde{X}\tilde{G}\tilde{X})\Pi = 0.\label{rare1}
\end{align}
Notice that $A_{\epsilon}$ and $\tilde{A}_{\epsilon}$ are related by
\begin{align}
\tilde{A}_{\epsilon}\Pi = \Pi A_{\epsilon} - & \Pi A_{\epsilon}\bar{\Pi}^{T}\bar{\Pi} - \epsilon \tilde{v}_{0}w_{0}^{T}\bar{\Pi}^{T}\bar{\Pi} \nonumber \\
& + \epsilon \Pi v_{0}w_{0}^{T} - \epsilon \Pi v_{0}\tilde{w}_{0}^{T}\Pi, \label{rare2}
\end{align}
where $\bar{\Pi}$ is the complement of $\Pi$. Combining (\ref{rare1}) and (\ref{rare2}) then yields
\begin{align}
A_{\epsilon}^{T}\hat{X} + \hat{X}A_{\epsilon} + Q - \hat{X}G\hat{X} = \mathcal{R}, \label{area}
\end{align}
with $\mathcal{R}$ denoting the residue of the equivalent ARE as
\begin{align}
\mathcal{R} := \bar{\Pi}^{T}\bar{\Pi}A_{\epsilon}^{T}\hat{X} + \hat{X}A_{\epsilon}\bar{\Pi}^{T}\bar{\Pi} + Q - \Pi^{T}\tilde{Q}\Pi.
\label{residue}
\end{align}
The expressions of the equivalent ARE above and its residue $\mathcal{R}$ are same as the ones in \cite{nan} (by replacing the notation $\Pi$ with $P$ in proof of Theorem 3.5). Note that when deriving the equivalent ARE above, different than \cite{nan}, the terms associated with $\epsilon$ in (\ref{rare2}) come up because of the consensus reformulation. These terms are eliminated by the relation $\hat{X}v_{0}=0$ as previously shown in Section \RNum{3}. The rest of the proof follows the same as in \cite{nan}.
\end{proof}

Facilitated by the preceding theorem, the objective function in (\ref{mainr}) then satisfies
\begin{align}
\| g_{\epsilon}(s) - \hat{g}_{\epsilon}(s)\|_{\mathcal{H}_{2},\bar{\omega}} \leq \gamma f(\xi).
\end{align}
Thereby, we can approach (\ref{mainr}) by minimizing its bound $f(\xi)$ following the same reason just explained. An important property of this bound is that the unknown $\Pi$ is only contained in the scalar $\xi$, and that $f(\xi)$ is monotonic in $\xi$. As a result, minimization of $f(\xi)$ is equivalent to that of the scalar $\xi$. This leads to the upper bound relaxation (\ref{xi}). Note that due to the nonconvex and combinatorial nature of the problem (\ref{mainr}), in general, it is impossible to quantify the optimality gap between (\ref{mainr}) and (\ref{xi}). The advantage of the relaxation (\ref{xi}) is that the monotonicity of $f(\xi)$ provides one possible way by which this optimality gap can at least be made small by minimizing $\xi$ to close to zero. 

\section*{Appendix C. Proofs}

\subsection{Lemma \ref{cont}}
We prove the controllability by contradiction. Suppose $(A,B)$ is not controllable, which according to PBH test is equivalent to the existence of a vector $v\neq 0$ such that $A^{T}v=\lambda v$ and $v^{T}B=0$. By partitioning $v$ equally as $v = \begin{bmatrix}
v_{1}^{T} & v_{2}^{T} & v_{3}^{T} & v_{4}^{T}
\end{bmatrix}^{T} $, we can write 
\begin{align*}
v^{T}B = \begin{bmatrix}
v_{1}^{T} & v_{2}^{T} & v_{3}^{T} & v_{4}^{T}
\end{bmatrix} \begin{bmatrix}
0 \\
0 \\
0 \\
B_{1}
\end{bmatrix} = v_{4}^{T}B_{1}.
\end{align*}
Given that $B_{1}\succ 0$, any $v$ satisfying the condition $v^{T}B=0$ must follow the form of $v = \begin{bmatrix}
v_{1}^{T} & v_{2}^{T} & v_{3}^{T} & 0
\end{bmatrix}^{T} $. On the other hand, $v$ also has to satisfy $v^{T}A = \lambda v^{T}$, which yields 
\begin{align*}
\begin{bmatrix}
L_{1m}^{T}v_{2} + L_{2m}^{T}v_{3} \\ v_{1} - D_{m}^{T}v_{2} \\ F_{1m}^{T}v_{2} + F_{2m}^{T}v_{3} \\ v_{3}
\end{bmatrix} & = \lambda \begin{bmatrix}
v_{1} \\ v_{2} \\ v_{3} \\ 0
\end{bmatrix}
\end{align*}
Since $F_{1m} = M^{-\frac{1}{2}}F_{1}M^{-\frac{1}{2}}$ is nonsingular, it can be easily verified that $v_{1} = v_{2} = v_{3} =0$ and then $v=0$. This contradicts $v\neq 0$, and thus proves that $(A,B)$ is controllable. The controllability of $(\tilde{A},\tilde{B})$ can be proven by the same rational given the fact that matrix $P$ is orthonormal and thus $\tilde{F}_{1} = PF_{1}P^{T}$ is nonsingular. 

\subsection{Lemma \ref{dclqr}} 

From the definition of $A_{\epsilon}$, the zero mode of the consensus point in $A$ is shifted to the left-half plane without change of basis. This makes $A_{\epsilon}$ Hurwitz, and makes $(A_{\epsilon},BR^{-\frac{1}{2}})$ and $(Q^{\frac{T}{2}},A_{\epsilon})$ trivially stabilizable and detectable. Thereby, the LQR problem $(A_{\epsilon},B,Q,R)$ satisfies the two existence conditions, and guarantees a unique stabilizing solution $X_{\epsilon}\succeq 0$. In addition, by assuming a non-zero vector $v\in ker(X_{\epsilon})$, pre- and post-multiplying (\ref{areep}) with $v$ yields $v^{T}Qv = 0$, while post-multiplying (\ref{areep}) with $v$ yields $X_{\epsilon}A_{\epsilon}v + Qv = 0$. These two equations imply that $ker(X_{\epsilon})$ is an $A_{\epsilon}$-invariant subspace contained in the null-space of $Q$. That is $X_{\epsilon}v_{0}=0$ given $A_{\epsilon}v_{0} = -\epsilon v_{0}$ and $Qv_{0}=0$. Consider a coordinate transformation on state matrices $A-GX_{\epsilon}$ and $A_{\epsilon}-GX_{\epsilon}$ as
\begin{align*}
V^{-1}(A_{\epsilon} - GX_{\epsilon})V = & \begin{bmatrix}
-\epsilon & -w_{0}^{T}GX_{\epsilon}v_{1} \\
0 & \Lambda_{1} - w_{1}^{T}GX_{\epsilon}v_{1}
\end{bmatrix}, \\
V^{-1}(A - GX_{\epsilon})V = & \begin{bmatrix}
0 & -w_{0}^{T}GX_{\epsilon}v_{1} \\
0 & \Lambda_{1} - w_{1}^{T}GX_{\epsilon}v_{1}
\end{bmatrix}.
\end{align*}
It can then be easily seen that $A-GX_{\epsilon}$ preserves the zero mode of the consensus point, and the rest of the eigenvalues are same to those of $\Lambda_{1} - w_{1}^{T}GX_{\epsilon}v_{1}$, which are independent of $\epsilon$ and are all on the left half plane given that $A_{\epsilon}-GX_{\epsilon}$ is Hurwitz. This completes the proof.

\subsection{Lemma \ref{psrp}}
Given $P$ defined over $w=\bar{v}$, it can be verified from Definition \ref{Pdef} that $P$ satisfies $P^{T}P \bar{v} = \bar{v}$ for any clustering set $\mathcal{I}$. As a result, we can write
\begin{align}
\tilde{L}_{1m} P \bar{v} = P L_{1m} P^{T} P \bar{v} = P L_{1m} \bar{v} =0,
\label{lapeig} 
\end{align}
which implies that $P \bar{v}$ is the right eigenvector of $\tilde{L}_{1m}$ corresponding to the zero eigenvalue. Following the same rationale as in (\ref{lapeig}), we can also show that $\tilde{L}_{2m} P \bar{v} = \tilde{L}_{3m} P \bar{v} = 0$. Therefore, denoting $\tilde{v}_{0} = \Pi v_{0} = \begin{bmatrix}
\bar{v}^{T}P^{T} & 0 & 0 & 0
\end{bmatrix}^{T}$, it holds that $\tilde{A}\tilde{v}_{0} = 0$, which means $\tilde{A}$ preserves the zero mode at $\tilde{v}_{0}$. The sign of $\tilde{Q}$ follows from eigenvalue's interlacing property, which guarantees $\tilde{Q}\succeq 0$ given that $\tilde{Q} = \Pi Q\Pi^{T}$ is congruent to $Q$. The null space of $\tilde{Q}$ can be similarly proven by $\tilde{Q}\tilde{v}_{0} = \Pi Qv_{0} = 0$.

\subsection{Lemma \ref{phiex}}

By definition of Hamiltonian matrix, the closed-loop state matrix $A_{\epsilon}-GX_{\epsilon}$ can be found from the Hamiltonian eigenspace as $A_{\epsilon}-GX_{\epsilon} = Z\Lambda Z^{-1}$ \cite{antoulas}. Using this expression, we can write $S(\bar{\omega})$ as
\begin{align}
S(\bar{\omega}) & = - \frac{\mathfrak{i}}{2\pi} \ln[(\mathfrak{i}\bar{\omega} I - A_{\epsilon} + GX_{\epsilon})(-\mathfrak{i}\bar{\omega} I - A_{\epsilon} + GX_{\epsilon})^{-1}] \nonumber \\
& = - \frac{\mathfrak{i}}{2\pi} Z \ln[(\mathfrak{i}\bar{\omega} I - \Lambda )(-\mathfrak{i}\bar{\omega} I - \Lambda )^{-1}] Z^{-1}.
\end{align}
According to \cite{antoulas}, the solution of the Lyapunov equation (\ref{lyap}) follows the form of (\ref{gramsolu}), where matrix $\mathcal{C}$ is initially defined by
\begin{align}
\mathcal{C}_{ij} = -\frac{[Z^{-1}S(\bar{\omega})B_{d}B_{d}^{*}Z^{-*} {+ } Z^{-1}B_{d}B_{d}^{*}S(\bar{\omega})Z^{-*}]_{ij}}{\lambda_{i} + \lambda_{j}^{*} } .
\end{align} 
Combining the two equations above, and after a few calculations yield the expression in (\ref{cauchyim}).

\subsection{Theorem \ref{kerror}}

To prove the error bound in (\ref{xierror}), we define an intermediate matrix $\bar{\Phi}$ as
\begin{align}
\bar{\Phi} =\begin{bmatrix} \bar{\Phi}_{1} & \cdots & \bar{\Phi}_{n_{d}}
\end{bmatrix},
\label{nphi}
\end{align}
where $\bar{\Phi}_{i} = Zdiag(Z^{-1}B_{di})\bar{\mathcal{C}}^{\frac{1}{2}}$, $i=1,...,n_{d}$, $B_{di}$ is the $i^{th}$ column of $B_{d}$, and $\bar{\mathcal{C}} = \bar{\mathcal{C}}^{\frac{1}{2}} \bar{\mathcal{C}}^{\frac{T}{2}}$ is defined by
\begin{align*}
\bar{\mathcal{C}}_{ij} = - \frac{c_{i}+c_{j}^{*}}{\lambda_{i}+\lambda_{j}^{*}}. 
\end{align*}
It can be verified that matrix $\bar{\Phi}$ satisfies $\Phi = \bar{\Phi}\bar{\Phi}^{T} = \Phi^{\frac{1}{2}}\Phi^{\frac{T}{2}}$. We further partition $\bar{\mathcal{C}}^{\frac{1}{2}}$ - the Cholesky decomposition of $\bar{\mathcal{C}}$ by
\begin{align*}
\bar{\mathcal{C}}^{\frac{1}{2}}  = \begin{bmatrix}
\bar{\mathcal{C}}^{\frac{1}{2}}_{11} & 0 \\
\bar{\mathcal{C}}^{\frac{1}{2}}_{21} & \bar{\mathcal{C}}^{\frac{1}{2}}_{22}
\end{bmatrix} = \begin{bmatrix}
\bar{\mathcal{C}}^{\frac{1}{2}}_{1{:}\kappa,1{:}\kappa} & 0 \\
\bar{\mathcal{C}}^{\frac{1}{2}}_{\kappa{+}1{:}n,1{:}\kappa} & \bar{\mathcal{C}}^{\frac{1}{2}}_{\kappa{+}1{:}n,\kappa{+}1{:}n}
\end{bmatrix}.
\end{align*}
With these notations, $\bar{\Phi}_{i}$ in (\ref{nphi}) can be decoupled into $ \bar{\Phi}_{i} = \bar{\Phi}_{i,s} + \bar{\Phi}_{i,f} $, where
\begin{align*}
\bar{\Phi}_{i,s} & = \Big[ [Z]_{:,1:\kappa}diag([Z^{-1}]_{1:\kappa,:}B_{di})\mathcal{C}^{\frac{1}{2}}_{11} \quad  0 \Big],  \\
\bar{\Phi}_{i,f} & = \Big[
[Z]_{:,\kappa+1:n}diag([Z^{-1}]_{\kappa+1:n,:}B_{di})\mathcal{C}^{\frac{1}{2}}_{21} \ \ ... \\
& \quad \quad \quad \quad \quad ... \quad  [Z]_{:,\kappa+1:n}diag([Z^{-1}]_{:,\kappa+1:n}B_{di})\mathcal{C}^{\frac{1}{2}}_{22}
\Big],
\end{align*}
and $\Phi_{\kappa}$ can be rewritten as $\Phi_{\kappa} = \sum_{i=1}^{n_{d}} \bar{\Phi}_{i,s}\bar{\Phi}_{i,s}^{T}$. Facilitated by this expression of $\Phi_{\kappa}$, we can find that $\xi = \| (I_{n}-\Pi^{T}\Pi) \Phi^{\frac{1}{2}} \|_{F} = \| (I_{n}-\Pi^{T}\Pi) \bar{\Phi}\|_{F}$ satisfies 
\begin{align}
\xi \leq \| (I_{n}-\Pi^{T}\Pi)\Phi_{\kappa}^{\frac{1}{2}} \|_{F} + \| (I_{n}-\Pi^{T}\Pi) \bar{\Phi}_{f} \|_{F},
\label{xi2}
\tag{58}
\end{align}
where $\bar{\Phi}_{f} = \begin{bmatrix}
\bar{\Phi}_{1,f} & \cdots & \bar{\Phi}_{n_{b},f}
\end{bmatrix}$. 
Notice that the second norm on the RHS of (\ref{xi2}) is further bounded by $\| (I_{n}-\Pi^{T}\Pi)\bar{\Phi}_{f} \|_{F} \leq \| \bar{\Phi}_{f}\|_{F}$ with %
\begin{small}
\begin{align*}
 \| \bar{\Phi}_{f}\|_{F} & \leq \sqrt{\eta^{2} \sum_{i=1}^{n_{d}} (\| \mathcal{C}^{\frac{1}{2}}_{21}\|_{F}^{2} + \| \mathcal{C}^{\frac{1}{2}}_{22}\|_{F}^{2}) } = \sqrt{\eta^{2} n_{d}\sum_{i=\kappa+1}^{4n} -\frac{c_{i}+c_{i}^{*}}{\lambda_{i}+\lambda_{i}^{*}}}.
\end{align*} %
\end{small}%
Inserting this along with $\Pi_{*}$ to the RHS of (\ref{xi2}) yields the error bound in (\ref{xierror}).

\end{document}